\documentclass[onecolumn,letterpaper, 10 pt,journal]{IEEEtran}

\IEEEoverridecommandlockouts
\usepackage{cite}
\usepackage{amsmath,amssymb,amsfonts}
\usepackage{amsthm}
\usepackage{algorithmic}
\usepackage{graphicx}
\usepackage{breqn}
\usepackage{booktabs}
\usepackage{textcomp}
\usepackage{xcolor}
\usepackage{soul}
\usepackage[caption=false]{subfig}

\theoremstyle{definition}
\newtheorem{definition}{Definition}
\newtheorem{property}{Property}
\newtheorem{theorem}{Theorem}
\newtheorem{proposition}{Proposition}
\newtheorem{assum}{Assumption}

\def\BibTeX{{\rm B\kern-.05em{\sc i\kern-.025em b}\kern-.08em
    T\kern-.1667em\lower.7ex\hbox{E}\kern-.125emX}}
\title{\LARGE \bf
	An incremental input-to-state stability condition for a generic class of recurrent neural networks}
\author{William D'Amico, Alessio La Bella, Marcello Farina%
\thanks{The authors are with Dipartimento di Elettronica, Informazione e Bioingegneria,
	Politecnico di Milano, Via Ponzio 34/5, 20133, Milano, Italy (e-mail: \text{william.damico@polimi.it},  \text{alessio.labella@polimi.it}, and \text{marcello.farina@polimi.it}). {Corresponding author: William D'Amico}.}
}
\begin{document}

\maketitle
\thispagestyle{empty}
\pagestyle{empty}

\begin{abstract}
This paper proposes a novel sufficient condition for the incremental input-to-state stability of a generic class of recurrent neural networks (RNNs). The established condition is compared with others available in the literature, showing to be less conservative. 
Moreover, it can be applied for the design of incremental input-to-state stable RNN-based control systems, resulting in a linear matrix inequality constraint for some specific RNN architectures. The formulation of nonlinear observers for the considered system class, as well as the design of control schemes with explicit integral action, are also investigated. The theoretical results are validated through simulation on a referenced nonlinear system.
\end{abstract}


Neural networks (NNs) have gained interest in many engineering fields, given the ever-growing availability of large amounts of data, e.g., collected measurements from plants, and their significant ability to reproduce nonlinear dynamics~\cite{10.1007/11840817_66,hornik1989multilayer}. In particular, NNs have shown to be particularly suited for control applications, as they can be used not only to identify unknown systems, but also to directly design feedback controllers from data~\cite{hunt1992neural,d2022recurrent}.
Among existing NN architectures, recurrent neural networks (RNNs) are typically adopted for controlling dynamical systems, since they are inherently characterized by the presence of state variables~\cite{bonassi2022recurrent}. 

Despite the increasing popularity of RNNs, their theoretical properties have been rarely analysed. As nonlinear dynamical systems, it is in fact fundamental to characterize conditions that guarantee the stability of their motions, especially when RNNs are part of control systems~\cite{barabanov2002stability,liu2021overview}. In this context \textit{incremental input-to-state stability} ($\delta$ISS)~\cite{angeli2002lyapunov,tran2016incremental} plays a crucial role. This property entails that, asymptotically, the state trajectories are solely determined by the applied inputs and not by their initial conditions~\cite{bayer2013discrete}. Thus, the dynamics of a $\delta$ISS RNN is independent of its initialization. The $\delta$ISS property also enables the design of trivial observers for the RNN states: the latter, indeed, can be asymptotically estimated just exploiting the knowledge of the applied inputs. Finally, note that the $\delta$ISS implies other common stability properties, e.g., global asymptotic stability (GAS) of the equilibria and input-to-state stability (ISS)~\cite{bayer2013discrete,angeli2002lyapunov}.

Motivated by this, the paper presents a novel $\delta$ISS sufficient condition for a generic class of RNN architectures. The proposed condition is applicable to control systems, and in particular for the analysis and design of RNNs-based feedback controllers and feedforward compensators. 

\subsection{State of the art and contribution}
Despite the large popularity and potentialities of RNNs in control applications, relatively few stability results have been discussed in the literature. Sufficient conditions ensuring stability-related properties for RNNs are presented in~\cite{miller2018stable} and in~\cite{stipanovic2021stability}, the latter focusing on a specific class of RNNs, i.e., gated recurrent units (GRU).
A stability condition for a generic class of RNNs is discussed in~\cite{hu2002global}, considering however the case with constant inputs.
Note that the above-mentioned contributions address stability properties weaker than the $\delta$ISS (e.g., the GAS property), which do not consider the effect of inputs \cite{angeli2002lyapunov}. This has motivated other research studies to focus on conditions guaranteeing $\delta$ISS. The latter, interestingly, can be directly enforced in the data-based RNN training phase, e.g., as discussed in~\cite{bonassi2021stability, bonassi2021stabilitygru}. However, these works focus on open-loop RNNs, and they do not address the design of stabilizing RNN-based feedback controllers.\\
Regarding control systems, in~\cite{yin2021stability} the stability is analysed in case of feedforward NN (FFNN) controllers and assuming a linear controlled system with uncertainties. Design conditions for FFNN controllers are also provided in~\cite{vance2008discrete}, considering specific classes of second-order nonlinear systems under control. 
Also model predictive control (MPC) has been investigated as a method for the design of efficient controllers applicable to systems described by specific classes of RNNs. For instance, the ISS of a MPC-controlled neural nonlinear autoregressive exogenous (NNARX) system is discussed in~\cite{seel2021neural}. Also, MPC regulation strategies for other RNN architectures are presented in~\cite{armenio2019model} and \cite{terzi2021learning}, ensuring closed-loop stability if the RNN-based model of the controlled system enjoys the $\delta$ISS property.

In this work we first derive a novel $\delta$ISS condition for a generic class of discrete-time nonlinear systems, which includes the one analysed in~\cite{sontag1992neural}, as well as different common RNN classes (e.g., echo state networks and NNARX). We prove that the proposed condition is less conservative than existing ones established in the past years for RNNs, e.g., \cite{bonassi2021stability,armenio2019model,hu2002global}. The established results turn out to be particularly suited for the design of feedback controllers and feedforward compensators. In particular, they allow us to enforce the $\delta$ISS property on control systems, also in case the controlled system does not enjoy the same property. Moreover, we show that, if specific RNN-based control architectures are considered, the design problem translates to a linear matrix inequality (LMI) problem, efficiently solvable by common solvers.

\subsection{Paper outline}
The paper is structured as follows. In Section \ref{sec:th}, the equations of the considered generic class of RNNs are introduced. The novel sufficient condition $\delta$ISS is stated in Section \ref{sec:th2} and it is compared with other existing conditions in the literature. In Section \ref{sec:inter}, the $\delta$ISS properties of feedforward and feedback interconnected RNNs are investigated, whereas Section \ref{sec:lmi} discusses in details the controller design with $\delta$ISS guarantees. Section \ref{sec:sim} shows the application of the theoretical results in this paper to a simulation example, whereas conclusions are drawn in Section~\ref{sec:conclusions}.

\subsection{Notation and basic definitions} Given a matrix $A$, its transpose is $A^T$, and the transpose of its inverse is $A^{-T}$. The entry in the $i$-th row and $j$-th column of a matrix $A$ is denoted as $a_{ij}$. $|A|$ denotes a matrix whose entries are $|a_{ij}|$, for all $i,j$.
The $i$-th entry of a vector $v$ is indicated as $v_i$. Given a symmetric matrix $P$, we use $P\succeq0$, $P\succ0$, $P\preceq0$, and $P\prec0$ to indicate that it is positive semidefinite, positive definite, negative semidefinite, and negative definite, respectively. $\lambda_{min}(P)$ and $\lambda_{max}(P)$ denote the minimum and maximum eigenvalues of a symmetric matrix $P$, respectively. $0_{n,m}$ denotes a zero matrix with $n$ rows and $m$ columns and $I_n$ is the identity matrix of dimension $n$. Given a sequence of square matrices $A_1,A_2,\dots,A_n$, $D=$ diag$(A_1,A_2,\dots,A_n)$ is a block diagonal matrix having $A_1,A_2,\dots,A_n$ as sub-matrices on the main-diagonal blocks. Moreover, $\|v\|=\sqrt{v^Tv}$ denotes the 2-norm of a column vector $v$ and $\|v\|_Q=\sqrt{v^TQv}$ denotes the weighted Euclidean norm of $v$, where $Q$ is a positive definite matrix. Given a sequence $\vec{u}=u(0),u(1),\dots$, we define its infinity norm as $\|\vec{u}\|_{\infty}=\sup_{k\in\mathbb{N}}\|u(k)\|$. Also, $id_{n}(\cdot)$ denotes a column vector of dimension $n$ with all elements equal to the identity function $id(\cdot)$. We introduce the following definition.

\begin{definition}[\!\!\cite{cao2004absolute}]
	A real function $g:\mathbb{R}\to \mathbb{R}$ is called globally Lipschitz continuous if there exist a constant $L_p\geq 0$ such that, for any $x,y\in\mathbb{R}$, it holds that
	\begin{equation*}
		|g(x)-g(y)|\leq L_p|x-y|\,.
	\end{equation*}
\end{definition}
The following property will be used later in the paper.
\begin{property}[\!\!\cite{armenio2019model}]
	\label{Prop1}
	Given two vectors $a, b\in\mathbb{R}^n$, it holds that $\|a+b\|^2\leq(1+\tau^2)\|a\|^2+\left(1+\frac{1}{\tau^2}\right)\|b\|^2$ for any scalar $\tau\neq0$.
\end{property}
We now consider a general discrete-time nonlinear system expressed as
\begin{equation}
	\label{DiscreteSys}
	x(k+1)=f^o(x(k),u(k))\,,
\end{equation}
where $f^o:\mathbb{X}\times\mathbb{U}\to \mathbb{X}$, $\mathbb{X}\subseteq\mathbb{R}^{n}$, $\mathbb{U}\subseteq\mathbb{R}^{m}$, $0_{n,1}\in\mathbb{X}$, and $0_{m,1}\in\mathbb{U}$. Moreover, $f^o(\cdot)$ is such that $f^o(0_{n,1},0_{m,1})=0_{n,1}$,
$k\in\mathbb{Z}_{\geq0}$ is the discrete-time index, $x(k)\in\mathbb{X}$ is the state of the system and $u(k)\in\mathbb{U}$ is the exogenous variable. The set of admissible input sequences $\vec{u}$ is denoted by $\mathcal{U}$. We indicate with $x(k,x_{0},\vec{u})$ the solution to the system \eqref{DiscreteSys} at time $k$ starting from the initial state $x_{0}\in\mathbb{X}$ with input sequence $\vec{u}\in\mathcal{U}$.\\
Now, we recall some useful notions for the following (see~\cite{bayer2013discrete}).

\begin{definition}[$\mathcal{K}$ function\cite{bayer2013discrete}]
	A continuous function \mbox{$\alpha: \mathbb{R}_{\geq0}\to\mathbb{R}_{\geq0}$} is a class $\mathcal{K}$ function if $\alpha(s)>0$ for all $s>0$, it is strictly increasing, and $\alpha(0)=0$.
\end{definition}

\begin{definition}[$\mathcal{K}_{\infty}$ function\cite{bayer2013discrete}]
	A continuous function \mbox{$\alpha: \mathbb{R}_{\geq0}\to\mathbb{R}_{\geq0}$} is a class $\mathcal{K}_{\infty}$ function if it is a class $\mathcal{K}$ function and $\alpha(s)\to\infty$ for $s\to\infty$.
\end{definition}

\begin{definition}[$\mathcal{KL}$ function\cite{bayer2013discrete}]
	A continuous function \mbox{$\beta: \mathbb{R}_{\geq0}\times\mathbb{Z}_{\geq0}\to\mathbb{R}_{\geq0}$} is a class $\mathcal{KL}$ function if $\beta(s,k)$ is a class $\mathcal{K}$ function with respect to $s$ for all $k$, it is strictly decreasing in $k$ for all $s>0$, and $\beta(s,k)\to0$ as $k\to\infty$ for all $s>0$.		
\end{definition}

\begin{definition}[$\delta$ISS\cite{bayer2013discrete}]
	\label{deltaISSdef}
	System \eqref{DiscreteSys} is called \textit{incrementally input-to-state stable} if there exists a function $\beta\in\mathcal{KL}$ and a function $\gamma\in \mathcal{K}_{\infty}$ such that for any $k\in\mathbb{Z}_{\geq0}$, any initial states $x_{01},x_{02}\in \mathbb{X}$, and any couple of input sequences $\vec{u}_{1},\vec{u}_{2}\in\mathcal{U}$, it holds that
	\begin{align*}
		\|x(k,x_{01},\vec{u}_{1})-x(k,x_{02},\vec{u}_{2})\|\leq\beta(\|x_{01}&-x_{02}\|,k)+\gamma(\|\vec{u}_{1}-\vec{u}_{2}\|_{\infty})\,,
	\end{align*}
	where $x(k, x_{0i}, \vec{u}_{i})$, $i=1,2$, is the state of system \eqref{DiscreteSys} at time step $k$, computed from the initial condition $x_{0i}\in\mathbb{X}$ and with the input trajectory $\vec{u}_{i}\in\mathcal{U}$.
\end{definition}

\begin{definition}[Dissipation-form $\delta$ISS Lyapunov function\cite{tran2016incremental}]
	A function \mbox{$V:\mathbb{X}\times\mathbb{X}\to\mathbb{R}_{\geq0}$} is called a dissipation-form $\delta$ISS-Lyapunov function for \eqref{DiscreteSys}, if there exist functions $\xi_{1},\xi_{2},\xi\in\mathcal{K}_{\infty}$ and $\sigma\in\mathcal{K}$ so that, for all $x_{1},x_{2}\in\mathbb{X}$ and $u_{1},u_{2}\in\mathbb{U}$, it holds that
	
	\begin{equation}
		\label{firstCond1}
		\xi_{1}(\|x_{1}-x_{2}\|)\leq V(x_1,x_2)\leq\xi_{2}(\|x_{1}-x_{2}\|)\,,
	\end{equation}
	\begin{equation}
		\label{secondCond1}
		\begin{split}
			V(f(x_{1},u_{1}),f(x_{2},u_{2}))-V(x_{1},x_{2}) \leq-\,\xi(\|x_{1}-x_{2}\|)+\,\sigma(\|u_{1}-u_{2}\|)\,.&
		\end{split}
	\end{equation}
\end{definition}

\begin{theorem}[\!\!\cite{tran2016incremental}]
	\label{deltaISS NC}
	If system \eqref{DiscreteSys} admits a dissipation-form $\delta$ISS Lyapunov function, then it is $\delta$ISS.
\end{theorem}

%
%
\section{Problem statement}
\label{sec:th}
We consider the following class of nonlinear discrete-time systems:
\begin{subequations}
	\label{nonlinclassall}
	\begin{align}
		\label{nonlinclass}
		x(k+1)&=f(Ax(k)+Bu(k))\,,\\
		\label{nonlinclassout}
		y(k)&=Cx(k)+Du(k)\,,
	\end{align}
\end{subequations}
where $u(k)\in\mathbb{R}^m$ is the exogenous variable, \mbox{$y(k)\in\mathbb{R}^l$} is the output vector, $x(k)\in\mathbb{R}^n$ is the state vector, $f(\cdot)=\begin{bmatrix}f_1(\cdot)&\dots&f_n(\cdot)\end{bmatrix}^T\in\mathbb{R}^n$ is a vector of scalar functions applied element-wise, $A\in\mathbb{R}^{n\times n}$, $B\in\mathbb{R}^{n\times m}$, \mbox{$C\in\mathbb{R}^{l\times n}$}, and $D\in\mathbb{R}^{l\times m}$. The exogenous variable $u(k)$ takes different roles in the various setups considered in this paper. Namely, $u(k)$ can be the manipulable input variable in case of open-loop systems, whereas it can be the output reference or the exogenous disturbance in case of closed-loop control systems. In this work, the function $f(\cdot)$ takes the particular form specified in the following assumption.
\begin{assum}
	\label{Ass1}
	The functions $f_i(\cdot)$, $i=1,\dots,n$, are nonlinear globally Lipschitz continuous functions with Lipschitz constant $L_{pi}$, or the identity function $id(\cdot)$.
\end{assum}
Given a system in class \eqref{nonlinclassall}, let us introduce the set $$\mathcal{L}\!=\!\{i \in\{1,\dots,n\}\,|\;f_i(\cdot)\!\neq\! id(\cdot)\}\,.$$ Note that, under Assumption \ref{Ass1}, system \eqref{nonlinclassall} is representative of several RNN architectures. For instance, \eqref{nonlinclassall} includes the general formulation of RNNs considered in~\cite{sontag1992neural}, where $D=0_{l,m}$, \mbox{$f_1(\cdot)=\dots=f_n(\cdot)=\sigma_f(\cdot)$}, and $\sigma_f:\mathbb{R}\to \mathbb{R}$ is a globally Lipschitz function (e.g., rectified linear unit (ReLU), sigmoid, or $\tanh$). Also, as better clarified below, many other RNNs considered in the literature can be written in form~\eqref{nonlinclassall}, possibly under some assumptions and/or minor reformulations. 
\subsection{Example 1: echo state networks (ESNs)}
\label{esn1}
ESNs \cite{jaeger2001echo} are particular types of RNNs composed of a dynamical reservoir (hidden layer) in which the connections between neurons are sparse and random. If we consider the formulation proposed in \cite{armenio2019model}, with $\nu$ neurons (i.e., states) in the reservoir, input $u(k)\in\mathbb{R}^{m}$, output $y(k)\in\mathbb{R}^l$, nonlinear Lipschitz continuous internal units output functions $\zeta(\cdot)$ applied element-wise, and linear output units output functions, the ESN equations are:
\begin{subequations}
	\label{esn}
	\begin{align}
		\chi(k+1) & = \zeta(W_x\chi(k)+W_u u(k)+W_yy(k))\,,
		\label{net_state}\\[1.5ex]
		y(k) & = W_{out_{1}}\chi(k)+W_{out_{2}}u(k-1)\,,
		\label{net_output}
	\end{align}
\end{subequations}
where $W_x\in\mathbb{R}^{\nu\times \nu}$, $W_{u}\in\mathbb{R}^{\nu\times m}$, $W_y\in\mathbb{R}^{\nu\times l}$, \mbox{$W_{out_{1}}\in\mathbb{R}^{l\times \nu}$}, and $W_{out_{2}}\in\mathbb{R}^{l\times m}$.\\
Note that model \eqref{esn} can be reformulated as \eqref{nonlinclassall} by defining  $x(k)=\begin{bmatrix}\chi(k)^T&z(k)^T\end{bmatrix}^T$, where $z(k)=u(k-1)$, and by setting
\begin{align}
	\label{Aesn}
	A=\left[\begin{matrix}
		W_x^*&W_yW_{out_{2}}\\
		0_{m,\nu}&0_{m,m}\end{matrix}\right],
\end{align}
$B=\begin{bmatrix}W_u^T&I_m\end{bmatrix}^T$, $C=\begin{bmatrix}W_{out_{1}}&W_{out_{2}}\end{bmatrix}$, $D=0_{l,m}$, and $f(\cdot)=\begin{bmatrix}\zeta(\cdot)^T&id_{m}(\cdot)^T\end{bmatrix}^T$, where $W_x^*=W_x+W_yW_{out_{1}}$.
%
%
\subsection{Example 2: shallow neural nonlinear autoregressive exogenous (NNARX) models}
NNARX is a class of nonlinear autoregressive exogenous models where a FFNN is used as nonlinear regression function. As shown in \cite{bonassi2021stability}, a shallow (i.e., 1-layer) NNARX, with input $\tilde{u}(k)\in\mathbb{R}^{\widetilde{m}}$, output $y(k)\in\mathbb{R}^l$, and $\nu$ neurons, is a dynamical system defined by the following equation
\begin{equation}
	\label{nnarx}
	y(k+1) = W_0\zeta(W_\phi \phi(k)+W_u\tilde{u}(k)+b)+b_0\,,
\end{equation}
where $W_0\in\mathbb{R}^{l\times \nu}$, $b_0\in\mathbb{R}^{l}$, the vector $\phi(k)\in\mathbb{R}^{(l+\widetilde{m})N}$ is defined as
\begin{align}
	\phi(k) = 
	\big[&\tilde{u}(k-N)^T,y(k-N+1)^T,...\,,\tilde{u}(k-2)^T,y(k-1)^T,\tilde{u}(k-1)^T,\,y(k)^T \,\big]^T\,,	\label{regressor}
\end{align}
$\zeta(\cdot)\in\mathbb{R}^\nu$ is a vector of Lipschitz continuous activation functions applied element-wise, $W_u\in\mathbb{R}^{\nu\times \widetilde{m}}$, $b\in\mathbb{R}^{\nu}$, $W_\phi=\begin{bmatrix}W_{\phi_1}&W_{\phi_2}&W_{\phi_3}\end{bmatrix}\in\mathbb{R}^{\nu\times (\widetilde{m}+l)N}$, $W_{\phi_1}\in\mathbb{R}^{\nu\times (\widetilde{m}+l)}$, $W_{\phi_2}\in\mathbb{R}^{\nu\times \tau}$, $\tau=(l+\widetilde{m})N-2l-\widetilde{m}$, and $W_{\phi_3}\in\mathbb{R}^{\nu\times l}$. 
\begin{proposition}
	\label{classnnarx}
	Model \eqref{nnarx} can be reformulated as \eqref{nonlinclassall} by defining 
	$u(k)=\begin{bmatrix}\tilde{u}(k)^T&1\end{bmatrix}^T$, \begin{align}\label{statennarx}
		x(k)=\big[\, \widetilde{\phi}(k)^T,\;\upsilon(k)^T \,\big]^T\,,
	\end{align} 
	where 
	\begin{align*}
		\widetilde{\phi}(k)&=\big[ \tilde{u}(k-N)^T,y(k-N+1)^T,...\,,\tilde{u}(k-2)^T,y(k-1)^T,\tilde{u}(k-1)^T\big]^T\,,\\
		\upsilon(k)&=\zeta(W_\phi \phi(k-1)+W_u\tilde{u}(k-1)+b)\,,
	\end{align*}
	and by setting
	\begin{align}
		\label{Annarx}
		A &=\left[\begin{matrix}
			0_{\tau,l+ \widetilde{m}}&I_{\tau}&0_{\tau,\nu}\\
			0_{l,l+ \widetilde{m}}&0_{l,\tau}&W_0\\
			0_{ \widetilde{m},l+ \widetilde{m}}&0_{ \widetilde{m},\tau}&0_{ \widetilde{m},\nu}\\
			W_{\phi_1}&W_{\phi_2}&W_{\phi_3}W_0
		\end{matrix}\right]\,,\\ B &=\left[\begin{matrix}
			0_{\tau, \widetilde{m}}&0_{\tau,1}\\
			0_{l, \widetilde{m}}&b_0\\	
			I_{ \widetilde{m}}&0_{ \widetilde{m},1}\\
			W_u&b+W_{\phi_3}b_0
		\end{matrix}\right]\,, \label{Bnnarx}
	\end{align} 
	$C=\begin{bmatrix}0_{l,n-\nu}&W_0\end{bmatrix}$, $D=\begin{bmatrix}0_{l, \widetilde{m}}&b_0\end{bmatrix}$, and \mbox{$f(\cdot)=\begin{bmatrix}id_{n-\nu}(\cdot)^T&\zeta(\cdot)^T\end{bmatrix}^T$}.
\end{proposition}
\begin{proof}
	Let us consider the extended input $u(k)$ and the extended state vector $x(k)\in\mathbb{R}^{n}$, with $n=\nu+(l+ \widetilde{m})N-l$. Firstly, note that
	\begin{align*}
		\widetilde{\phi}(k\!+\!1)\!=\!\left[\begin{matrix}
			0_{\tau,l+ \widetilde{m}}\!&\!I_{\tau}\!&\!0_{\tau,\nu}\\
			0_{l,l+ \widetilde{m}}\!&\!0_{l,\tau}\!&\!W_0\\
			0_{ \widetilde{m},l+ \widetilde{m}}\!&\!0_{ \widetilde{m},\tau}\!&\!0_{ \widetilde{m},\nu}\end{matrix}\right]\!x(k)\!+\!\left[\begin{matrix}
			0_{\tau, \widetilde{m}}\!&\!0_{\tau,1}\\
			0_{l, \widetilde{m}}\!&\!b_0\\	
			I_{ \widetilde{m}}\!&\!0_{ \widetilde{m},1}
		\end{matrix}\right]\!u(k). 
	\end{align*}
	Secondly,
	\begin{align*}
		\upsilon(k+1)&=\zeta(W_\phi \phi(k)+W_u\tilde{u}(k)+b)=\\
		&=\zeta(\begin{bmatrix}W_{\phi_1}\!&\!W_{\phi_2}\end{bmatrix}\!\widetilde{\phi}(k)\!+\!W_{\phi_3}y(k)+W_u\tilde{u}(k)+b)=\\
		&=\zeta(\begin{bmatrix}W_{\phi_1}\!&\!W_{\phi_2}\end{bmatrix}\!\widetilde{\phi}(k)\!+\!W_{\phi_3}\!W_0\upsilon(k)\!+\!W_{\phi_3}\!b_0\!+\!W_u\tilde{u}(k)\!+\!b)
	\end{align*}
	This concludes the proof of the statement.
\end{proof}

\subsection{Example 3: class of RNN systems in \cite{hu2002global}}
\label{eschina}
In \cite{hu2002global} a slightly different RNN class is considered, i.e.,
\begin{align}\label{eq:system_china1}
	x(k+1)=Ex(k)+O\, f(\hat{A}x(k)+s)\,,
\end{align}
where $\hat{A}$ is a full matrix, $E=\text{diag}(e_1,\dots,e_n)$, \mbox{$O=\text{diag}(o_1,\dots,o_n)$}, $s=\begin{bmatrix}s_1&\dots&s_n\end{bmatrix}^T$ is a vector of constant inputs, with $e_i,\,o_i,\,s_i \in \mathbb{R}$, $\forall i = 1,\dots,n$. Moreover,  \mbox{$f(\cdot)=\begin{bmatrix}f_1(\cdot)&\dots&f_n(\cdot)\end{bmatrix}^T$}, where each $f_i(\cdot)$ is a globally Lipschitz continuous activation function with Lipschitz constant $L_{pi}$. Note that, in case $E=0_{n,n}$ and $o_i = 1$ for all $i \in \mathcal{L}$, system \eqref{eq:system_china1} is in the same form of \eqref{nonlinclass}.

\medskip
Given the generic class of systems \eqref{nonlinclassall} under Assumption \ref{Ass1}, a sufficient  condition ensuring the $\delta$ISS property is established and described in the following section.

\section{A novel sufficient condition for incremental input-to-state stability of RNNs}
\label{sec:th2}
Let us consider a generic system \eqref{nonlinclassall} fulfilling Assumption~\ref{Ass1}. The following theorem provides a sufficient condition which guarantees the $\delta$ISS for nonlinear systems lying in the class \eqref{nonlinclass}. We first define a diagonal matrix $W=\text{diag}(L_{p1},\dots,L_{pn})\in\mathbb{R}^{n\times n}$, where $L_{pi}=1$ for all $i \notin \mathcal{L}$. We introduce the matrices $\widetilde{A}=WA$ and $\widetilde{B}=WB$.
\begin{theorem}
	\label{th2}
	Let Assumption \ref{Ass1} hold. System \eqref{nonlinclass} is $\delta$ISS if $\exists P=P^T\succ0$ such that $p_{ij}=p_{ji}=0$ $\forall i\!\in\!\mathcal{L}$, $ \forall j \!\in\! \{1,\dots,n\}$ with $ j\neq i$, and
	\begin{align}
		\label{lyaplin}
		\widetilde{A}^TP\widetilde{A}-P\prec0\,.
	\end{align}
\end{theorem}
\begin{proof}
	In order to prove the $\delta$ISS of system \eqref{nonlinclass} we show the existence of a dissipation-form $\delta$ISS Lyapunov function. We consider, as candidate, $V(x_1(k),x_2(k))=\|x_1(k)-x_2(k)\|^2_P$.
	
	From now on, for notational simplicity, we drop the dependence on $k$. If we consider the $\mathcal{K}_\infty$ functions $\xi_1(\|x_1-x_2\|)=\lambda_{min}(P)\|x_1-x_2\|^2$ and $\xi_2(\|x_1-x_2\|)=\lambda_{max}(P)\|x_1-x_2\|^2$, condition \eqref{firstCond1} is easily verified. To prove that $V(x_1,x_2)$ satisfies also condition \eqref{secondCond1}, we introduce the following notation: $v_1=Ax_1+Bu_1$,  $x_1^+=f(v_1)$, $v_2=Ax_2+Bu_2$, $x_2^+=f(v_2)$, $\delta x=x_1-x_2$, $\delta x^+=x_1^+-x_2^+$, and $\delta u=u_1-u_2$. We can write
	\begin{align}\delta x^+&=f(v_1)-f(v_2)=\notag\\
		&= Wv_1-Wv_2 +  f(v_1)-Wv_1-f(v_2)+Wv_2=\notag\\
		&=\widetilde{A}\delta x + \widetilde{B}\delta u \,\,+ \Delta \label{eq:diff}\,,
	\end{align}
	where $\Delta=f(v_1)-Wv_1-f(v_2)+Wv_2$. From \eqref{eq:diff}, we obtain that
	\begin{align}V(x_1^+,x_2^+)-V(x_1,x_2)&=(\delta x^+)^TP\delta x^+-\delta x^TP\delta x=\notag\\
		&=(\widetilde{A}\delta x + \widetilde{B}\delta u)^TP(\widetilde{A}\delta x + \widetilde{B}\delta u)+2(\widetilde{A}\delta x + \widetilde{B}\delta u)^TP\Delta\,+\Delta^TP\Delta-\delta x^TP\delta x \,.\label{eq:lya}
	\end{align}
	Now, we can observe that
	\begin{align}2(\widetilde{A}\delta x + \widetilde{B}\delta  u)^TP\Delta+\Delta^TP\Delta&=2(W(v_1-v_2))^TP\Delta+\Delta^TP\Delta=\notag\\
		&=(2Wv_1-2Wv_2+\Delta)^TP\Delta=\notag\\
		&=q^TPr\,, \label{eq:ps}
	\end{align} 
	where $q, r \in\mathbb{R}^n$, with $ q=Wv_1+f(v_1)-Wv_2-f(v_2)$ and $r=f(v_1)-Wv_1-f(v_2)+Wv_2$. The elements of the vectors $q$ and $r$ are
	\begin{equation*}
		\centering
		q_i=\left\{\begin{split}
			& f_i(v_{1i})-f_i(v_{2i})+L_{pi}(v_{1i}-v_{2i}) \hspace{5mm} &\text{if}\;
			i \in \mathcal{L}\,,\\
			& 2v_{1i}-2v_{2i} &\text{if}\;
			i \notin \mathcal{L}\,,\\
		\end{split}\right.
	\end{equation*}
	\begin{equation*}
		\centering
		r_i=\left\{\begin{split}
			& f_i(v_{1i})-f_i(v_{2i})-L_{pi}(v_{1i}-v_{2i}) \hspace{5mm} &\text{if}\;
			i \in \mathcal{L}\,,\\
			& 0 &\text{if}\;
			i \notin \mathcal{L}\,.\\
		\end{split}\right.
	\end{equation*}
	Therefore, by setting $p_{ij}=p_{ji}=0$ $\forall i\!\in\!\mathcal{L}$ and $ \forall j \!\in\! \{1,\dots,n\}$ with $ j\neq i$, we can compute from \eqref{eq:ps} that
	\begin{align}q^TPr=\!\sum_{i\in\mathcal{L}}p_{ii}\left((f_i(v_{1i})-f_i(v_{2i}))^2-\!L_{pi}^2(v_{1i}-v_{2i})^2\right)\leq0\,. \label{eq:ps2}
	\end{align} 
	Inequality~\eqref{eq:ps2} holds since, in view of Assumption \ref{Ass1}, $f_i(\cdot)$ are globally Lipschitz continuous functions, for all $i\in\mathcal{L}$, and $p_{ii}>0$ $\forall i$ since \mbox{$P=P^T\succ0$}.
	
	As a result, by exploiting \eqref{eq:ps2} and in view of Property \ref{Prop1}, for any $\tau\neq0$, we can write that
	\begin{align*}V(x_1^+,x_2^+)-V(x_1,x_2)&\leq(\widetilde{A}\delta x + \widetilde{B}\delta u)^TP(\widetilde{A}\delta x + \widetilde{B}\delta u)-\delta x^TP\delta x\leq\\
		&\leq(1+\tau^2)(\widetilde{A}\delta x)^TP\widetilde{A}\delta x+\left(1+\frac{1}{\tau^2}\right)(\widetilde{B}\delta u)^TP\widetilde{B}\delta u-\delta x^TP\delta x\leq\\
		&\leq -\lambda_{min}(A^*)\|x_1-x_2\|^2+\lambda_{u}\|u_1-u_2\|^2\,,\end{align*} 
	for any $\lambda_u>\lambda_{max}(B^*)$, where $A^*=P-(1+\tau^2)\widetilde{A}^TP\widetilde{A}\succ0$ in view of \eqref{lyaplin} for a sufficiently small $\tau$, and $B^*=\left(1+\frac{1}{\tau^2}\right)\widetilde{B}^TP\widetilde{B}\succeq0$.
	
	Finally, note that $\sigma(\|u_1-u_2\|)=\lambda_u\|u_1-u_2 \|^2$ is a $\mathcal{K}$ function, whereas $\xi(\|x_{1}-x_{2}\|)=\lambda_{min}(A^*)\|x_1-x_2\|^2$ is a $\mathcal{K}_{\infty}$ function. This concludes the proof.
\end{proof}
In short, Theorem \ref{th2} ensures that system \eqref{nonlinclass} is $\delta$ISS if $\widetilde{A}$ is Schur stable and if there exists a matrix $P$ with a specific structure fulfilling the Lyapunov inequality \eqref{lyaplin}. In particular, $P$ must have zero elements along all the rows and columns (except for the diagonal element) corresponding to the rows of \eqref{nonlinclass} whose activation function is nonlinear. 

The $\delta$ISS condition in Theorem \ref{th2} is now compared with other existing conditions for the RNN systems introduced in Section \ref{sec:th} to show its generality. 

\subsection{Comparison with the stability condition in \cite{armenio2019model} for ESNs}
Note that, in case of ESNs, the assumptions of Theorem~\ref{th2} require the existence of a positive definite block diagonal matrix $P=\text{diag}(P_1,P_2)$ fulfilling \eqref{lyaplin}, where $P_1\in\mathbb{R}^{\nu\times \nu}$ is a diagonal matrix, as $\mathcal{L}=\{1,\dots,\nu\}$, and $P_2\in\mathbb{R}^{m\times m}$ is a full symmetric matrix. On the other hand, in \cite[Proposition~1]{armenio2019model} the following sufficient condition for $\delta$ISS of system \eqref{esn} in state-space is proposed in case $\zeta(\cdot)$ has activation functions with Lipschitz constant $L_p\leq1$.
\begin{proposition}[\!\!\cite{armenio2019model}]
	\label{deltaISSESN}
	If $\|W_x^*\|<1$, system \eqref{esn} in the state-space form in \cite{armenio2019model} is $\delta$ISS.
\end{proposition}
Below we show that the assumptions of Theorem \ref{th2} are less conservative than the one of Proposition \ref{deltaISSESN}. 
\begin{proposition}
	\label{Prop4}
	Let $\zeta(\cdot)$ be a vector of sigmoid Lipschitz continuous functions with Lipschitz constant $L_p\leq1$. The set of systems \eqref{esn} satisfying the assumption of Proposition \ref{deltaISSESN} is a subset of the set of systems \eqref{esn} satisfying the assumptions of Theorem \ref{th2}.
\end{proposition}
\begin{proof}
	See the Appendix.
\end{proof}

We also show an example in which the assumption of Proposition \ref{deltaISSESN} is not fulfilled whereas the assumptions of Theorem~\ref{th2}~are. Let $\nu=2$, $m=1$, $l=1$, $$W_x=\left[\begin{matrix}
	0.8257&-0.4711\\
	-1.0149&0.137
\end{matrix}\right]\,,\ \ \ W_y=\left[\begin{matrix}
	-0.2919\\
	0.3018
\end{matrix}\right]\,,$$ $W_{out_{1}}=\begin{bmatrix}0.3999&-0.93\end{bmatrix}$, and $W_{out_{2}}=-0.1768$. Then, $\|W_x^*\|=1.1413>1$. With $A$ defined as in \eqref{Aesn} and with the choice $$P=\left[\begin{matrix}
	2.1444&0&0\\
	0&0.7221&0\\
	0&0&1.0254\\
\end{matrix}\right]\,,$$ we have that the maximum eigenvalue of $A^TPA-P$ is equal to $-0.3197$, and thus $A^TPA-P\prec0$.

\subsection{Comparison with the stability condition in \cite{bonassi2021stability} for shallow NNARX models}

In case of shallow NNARXs, the assumptions of Theorem \ref{th2} require the existence of a symmetric positive definite matrix $P$ fulfilling \eqref{lyaplin} such that $p_{ij}=p_{ji}=0$, \mbox{$\forall i \in \mathcal{L}=\{n-\nu+1,\dots,n\}$} and $\forall j \in \{1,\dots,n\}$ with $j\neq i$.

In \cite[Theorem 8]{bonassi2021stability} a sufficient condition for $\delta$ISS of a deep (i.e., $M$-layered) NNARX is proposed. For comparison purposes, we recall here the condition for a 1-layer NNARX~\eqref{nnarx}, as this belongs to the class \eqref{nonlinclassall}, where $\zeta(\cdot)$ has activation functions with Lipschitz constant $L_p$. Note also that, in~\cite{bonassi2021stability}, a slightly different state-space formulation with respect to the one in this paper is considered.
\begin{proposition}[\!\!\cite{bonassi2021stability}]
	\label{deltaISSNNARX}
	If $\|W_0\|\|W_\phi\|<\frac{1}{L_p\sqrt{N}}$, system \eqref{nnarx} in the state-space form in \cite{bonassi2021stability} is $\delta$ISS.
\end{proposition}
Below we show that the assumptions of Theorem \ref{th2} are less conservative than the one of Proposition \ref{deltaISSNNARX} for a shallow NNARX. 
\begin{proposition}\label{Prop6}
	Let $\zeta(\cdot)$ be a vector of nonlinear Lipschitz continuous functions with Lipschitz constant $L_p$. The set of systems \eqref{nnarx} satisfying the assumption of Proposition \ref{deltaISSNNARX} is a subset of the set of systems \eqref{nnarx} satisfying the assumptions of Theorem \ref{th2}.
\end{proposition}
\begin{proof}
	See the Appendix.
\end{proof}

	
Here we show an example in which the assumption of Proposition~\ref{deltaISSNNARX} is not fulfilled whereas the assumptions of Theorem~\ref{th2} hold. Let $l=1$, $\widetilde{m}=1$, $N=2$, $\nu=1$, \mbox{$L_p=1$}, $W_\phi=\left[\begin{matrix}
	-0.2130&-0.8657&-1.0431&-0.2701
\end{matrix}\right]$ and $W_0=0.6293$. Then, \mbox{$\|W_0\|\|W_\phi\|=0.8801>0.7071=\frac{1}{\sqrt{N}}$}. With the choice $$P=\left[\begin{matrix}
	0.8278&0.0095&0.1847&0\\
	0.0095&1.2258&0.7531&0\\
	0.1847&0.7531&2.5870&0\\
	0&0&0&0.8723
\end{matrix}\right]\,,$$ we have that the maximum eigenvalue of $\widetilde{A}^TP\widetilde{A}-P$ is equal to $-0.239$ and, thus, $\widetilde{A}^TP\widetilde{A}-P\prec0$.

\subsection{Comparison with the stability condition in \cite{hu2002global} for a class of RNN systems}
The work in \cite{hu2002global} analyses the stability properties of a slightly different RNN class represented by \eqref{eq:system_china1}. In \cite[Corollary~1]{hu2002global} some sufficient conditions for global exponential stability are proposed. Among them, the following condition has a similar structure to the one of Theorem \ref{th2}, and it is thus compared. 

\begin{proposition}[\!\!\cite{hu2002global}]\label{prop:china}
	System \eqref{eq:system_china1} is globally exponentially stable for any input $s$ if there exists a matrix $P\!=\!P^T \!\succ\! 0$ such that
	\begin{align}
		\label{chinacond}
		\hat{M}^T P \hat{M} - P \prec 0\,,
	\end{align}
	where $\hat{M}=|E|\, +\, W|O||\hat{A}|$, and $W=\text{diag}(L_{p1},\dots,L_{pn})$.
\end{proposition}

For the sake of comparison, we now show that the assumptions of Theorem \ref{th2} are less conservative than the one of Proposition \ref{prop:china} for the class of RNN systems \eqref{eq:system_china1} with $E=0_{n,n}$ and $o_i = 1$ $\forall i \in \mathcal{L}$. Therefore, the following proposition is stated.

\begin{proposition}\label{prop:china2}
	The set of systems \eqref{eq:system_china1} satisfying the assumption of Proposition \ref{prop:china} is a subset of the set of systems \eqref{eq:system_china1} satisfying the assumptions of Theorem \ref{th2} if $E=0_{n,n}$ and $o_i = 1$ $\forall i \in \mathcal{L}$.
\end{proposition}
\begin{proof}
	See the Appendix.
\end{proof}
Here we show an example in which the assumption of Proposition \ref{prop:china} is not fulfilled whereas the ones of Theorem~\ref{th2}~hold. Let $n=2$, $E=0_{2,2}$, $O=I_2$, $W=I_2$, $\mathcal{L}=\{1,2\}$, $$\hat{A}=\left[\begin{matrix}
	0.4178&-0.8544\\0.8199&0.3573
\end{matrix}\right].$$ Hence, the spectral radius of $\hat{M}$ is equal to $1.225$ and so \eqref{chinacond} is never fulfilled. With the choice $$P=\left[\begin{matrix}
	1.2122&0\\
	0&1.2657
\end{matrix}\right]\,,$$ we have that the maximum eigenvalue of $\widetilde{A}^TP\widetilde{A}-P$ is equal to $-0.1135$ and, thus, $\widetilde{A}^TP\widetilde{A}-P\prec0$.
	
\section{$\delta$ISS of interconnected RNNs}
\label{sec:inter}
We now show how Theorem \ref{th2} also applies to the interconnection of systems in class \eqref{nonlinclassall}.
\subsection{Feedforward interconnection}
In this section we investigate the $\delta$ISS conditions of the series of $M_s$ systems lying in the class \eqref{nonlinclassall}. Specifically, each system is numbered in increasing order with respect to $i$ and defined by
\begin{subequations}
	\label{cascall}
	\begin{align}
		\label{cascx}
		x_i(k+1)&=f_{s_i}(A_ix_i(k)+B_iu_i(k))\,, \\
		\label{cascy}
		y_i(k)&=C_ix_i(k)+D_iu_i(k)\,,
	\end{align}
\end{subequations}
where $u_i(k)\in\mathbb{R}^{m_i}$, $y_i(k)\in\mathbb{R}^{l_i}$, and $x_i(k)\in\mathbb{R}^{n_i}$. Because of the series interconnection, it holds that $u_i(k)=y_{i-1}(k)$, for all \mbox{$i=2,...,M_s$}, and the following assumption is required to have dimensional consistency,
\begin{assum}
	\label{Ass2}
	We assume that the dimension of $u_i(k)$ is the same of $y_{i-1}(k)$, i.e., $m_i=l_{i-1}$, for all \mbox{$i=2,...,M_s$}. 
\end{assum}
We can state the following result.
\begin{proposition}
	\label{ffp}
	Let Assumption \ref{Ass2} hold. The series of $M_s$ systems in the class \eqref{nonlinclassall} lies in the class \eqref{nonlinclassall}.
\end{proposition}
\begin{proof}
	Firstly, note that the input of the first subsystem is the input of the overall series interconnection, i.e., $u(k)=u_1(k)$, whereas the output of the last subsystem of the series is the overall output, i.e., $y(k)=y_{M_s}(k)$. Since $y_i(k)=u_{i-1}(k)$, for all $i=2,...,M_s$, due to the series interconnection and under Assumption \ref{Ass2}, the second subsystem can be written as
	\begin{align}
		\label{cascx22}
		\begin{split}
			x_2(k+1)&=f_{s_2}(A_2x_2(k)+B_2C_1x_1(k)+B_2D_1u(k)) \,,\\
			y_2(k)&=C_2 x_2(k)+D_2 C_1 x_1(k)+D_2 D_1 u(k)\,.
		\end{split}
	\end{align}
	Following the same reasoning, for $i=3,...,M_s$, it holds that
	\begin{subequations}
		\label{cascall2}
		\begin{align}
			\label{cascx2}
			x_i(k+1)&=f_{s_i}\big(B_i(\sum_{h=1}^{i-2}(\prod_{j=0}^{i-h-2}D_{i-1-j})C_hx_h(k))+B_iC_{i-1}x_{i-1}(k)+A_ix_i(k)+B_i(\prod_{j=0}^{i-2}D_{i-1-j})u(k))\,, \\
			\label{cascy2}
			y_i(k)&=\sum_{h=1}^{i-1}(\prod_{j=0}^{i-h-1}D_{i-j})C_hx_h(k)+C_ix_i(k)+(\prod_{j=0}^{i-1}D_{i-j})u(k)\,.
		\end{align}
	\end{subequations}
	From \eqref{cascall}-\eqref{cascx22}-\eqref{cascall2}, by introducing the extended state vector $x(k)=\begin{bmatrix}x_1(k)^T&\dots&x_{M_s}(k)^T\end{bmatrix}^T$, it is possible to construct the matrices \mbox{$A$, $B$, $C$, $D$} and the vector \mbox{$f(\cdot)=\begin{bmatrix}f_{s_1}(\cdot)^T&\dots&f_{s_{M_s}}(\cdot)^T\end{bmatrix}^T$} of the overall system in the form \eqref{nonlinclassall}. This concludes the proof.
\end{proof}

To guarantee the $\delta$ISS of the series of $M_s$ systems in the form \eqref{nonlinclassall}, one way is to write the overall series system as \eqref{nonlinclassall}, and then to impose the sufficient condition for $\delta$ISS in Theorem \ref{th2} to the overall system. Alternatively, we can impose the sufficient condition for $\delta$ISS in Theorem \ref{th2} to each subsystem. In \cite[Proposition 4.7]{angeli2002lyapunov}, a theoretical result proves that the series interconnection of two $\delta$ISS continuous-time systems is $\delta$ISS. The same property can be extended to discrete-time systems. In the following Proposition~\ref{cas} we show that, given a series of discrete-time systems each one satisfying the assumptions of Theorem \ref{th2} for $\delta$ISS, their series interconnection satisfies the same assumptions.
\begin{proposition}
	\label{cas}
	Let us consider a series of systems in the class~\eqref{nonlinclassall}, each one fulfilling the assumptions of Theorem~\ref{th2}. The series of these systems satisfies the assumptions of Theorem~\ref{th2}.
\end{proposition}
\begin{proof}
	See the Appendix.
\end{proof}
\subsection{Feedback interconnection}

\begin{figure}[t]
	\centering
	\includegraphics[width=0.5\columnwidth]{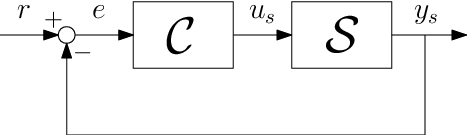}
	\caption{Feedback control scheme: $\mathcal{C}$ is the controller, $\mathcal{S}$ is the system to be controlled, $r$ is the reference signal, $e$ is the tracking error, $u_s$ is the manipulated variable, and $y_s$ is the output of the system.}
	\label{fig: CS}
\end{figure}

In this section we investigate the $\delta$ISS conditions of the feedback of two systems lying in the class \eqref{nonlinclassall}. The baseline feedback control scheme is depicted in Figure \ref{fig: CS}. We define the equations of the controller $\mathcal{C}$ as
\begin{subequations}
	\label{contall}
	\begin{align}
		\label{contx}
		x_c(k+1)&=f_{c}(A_cx_c(k)+B_cu_c(k))\,, \\
		\label{conty}
		y_c(k)&=C_cx_c(k)+D_cu_c(k)\,,
	\end{align}
\end{subequations}
where $u_c(k)\in\mathbb{R}^{m_c}$, $y_c(k)\in\mathbb{R}^{l_c}$, and $x_c(k)\in\mathbb{R}^{n_c}$. We consider the case in which the system is strictly proper to avoid algebraic loops. Thus, we define the equations of the controlled system $\mathcal{S}$ as
\begin{subequations}
	\label{sisall}
	\begin{align}
		\label{sisx}
		x_s(k+1)&=f_{s}(A_sx_s(k)+B_su_s(k))\,,\\
		\label{sisy}
		y_s(k)&=C_sx_s(k)\,,
	\end{align}
\end{subequations}
where $u_s(k)\in\mathbb{R}^{m_s}$, $y_s(k)\in\mathbb{R}^{l_s}$, and $x_s(k)\in\mathbb{R}^{n_s}$. The following assumption is required to have dimensional consistency.
\begin{assum}
	\label{Ass3}
	We assume that the dimension of $u_s(k)$ is the same of $y_c(k)$, and the one of $u_c(k)$ is the same of $y_s(k)$, i.e., $m_s=l_c$ and $m_c=l_s$. 
\end{assum}
We can state the following result. 
\begin{proposition}
	\label{fbp}
	Let Assumption \ref{Ass3} hold. The feedback interconnection in Figure \ref{fig: CS} of the systems \eqref{contall}-\eqref{sisall} in the class \eqref{nonlinclassall} lies in the class \eqref{nonlinclassall}. 
\end{proposition}
\begin{proof}
	Firstly, note from Figure \ref{fig: CS} that $r(k)$ is overall input to the feedback interconnection, i.e., $u(k)=r(k)$, whereas $y_s(k)$ is the overall output, i.e., $y(k)=y_s(k)$. Under Assumption~\ref{Ass3}, since $u_c(k)=e(k)=r(k)-y_s(k)=r(k)-C_sx_s(k)$ due to the negative feedback interconnection, and \mbox{$u_s(k)=y_c(k)=C_cx_c(k)+D_cu_c(k)$}, we can write the equations of the overall closed-loop system in the formulation \eqref{nonlinclassall} through the following definitions:
	\begin{align*}
		A =\left[\begin{matrix}
			A_c&-B_c C_s\\
			B_s C_c&A_s-B_s D_c C_s
		\end{matrix}\right]\,, \ \ \ B =\left[\begin{matrix}
			B_c\\
			B_s D_c
		\end{matrix}\right]\,,
	\end{align*} 
	$C=\begin{bmatrix}0_{l_s,n_c}&C_s\end{bmatrix}$, $D=0_{l_s,l_s}$, $x(k)=\begin{bmatrix}x_c(k)^T&x_s(k)^T\end{bmatrix}^T$, and $f(\cdot)=\begin{bmatrix}f_{c}(\cdot)^T&f_{s}(\cdot)^T\end{bmatrix}^T$.
\end{proof}
It is possible to show that the previous result holds also for the case in which the controller is strictly proper and the system is not. We omit the proof for the sake of conciseness.

In view of Proposition~\ref{fbp} and Theorem~\ref{th2}, it is possible to analyse or enforce (through the tuning of the parameters of $\mathcal{C}$) the $\delta$ISS property to the feedback control scheme in Figure~\ref{fig: CS}, where both the controller and the system are in the class \eqref{nonlinclassall}.

\section{Controller design with $\delta$ISS guarantees}
\label{sec:lmi}
This section discusses the design of controllers that confer $\delta$ISS guarantees to the control system. In this paper, we will not focus on the performances of the control system, which will be a matter of future research. 

In general, the $\delta$ISS condition \eqref{lyaplin} in Theorem~\ref{th2} corresponds to a nonlinear constraint in control design, which can be handled by common nonlinear solvers. On the other hand, there are some particular cases in which it can be reformulated as a linear matrix inequality (LMI) constraint, as shown in the following section.

\subsection{LMI-based control design}
We consider a control system whose overall equations are in the class \eqref{nonlinclassall}. In this section we show that, in some particular cases, the matrix $A$ of the closed-loop system can be written as $A=F+GJ$, or as $A=F+JG$, where $F$ and $G$ are known matrices depending on the system to be controlled, and $J$ is a matrix to be tuned taking the role of the control gain. Therefore, the objective is to tune $J$ so that the closed-loop system enjoys the $\delta$ISS property. In this respect, the following results hold.
\begin{proposition}
	\label{lmicl}
	Let us consider a system with equations in the class \eqref{nonlinclassall}, where $A=F+GJ$. Let $\widetilde{F}=WF$ and \mbox{$\widetilde{G}=WG$}, where $W=\text{diag}(L_{p1},\dots,L_{pn})$ is the diagonal matrix defined in Section \ref{sec:th2}. If $\exists P=P^T$, having the structure required by Theorem \ref{th2}, and $\exists H$ such that
	\begin{align}
		\label{lmi1}
		\begin{bmatrix}
			P & (\widetilde{F} P+ \widetilde{G}H)^T \\
		\widetilde{F}P+\widetilde{G}H & P
		\end{bmatrix}\succ0\,,
	\end{align}
	then, if we set $J=HP^{-1}$, the system \eqref{nonlinclassall} is $\delta$ISS.
\end{proposition}
\begin{proof}
	Firstly, note that from \eqref{lmi1} it follows that $P\succ0$ and $P^{-1}\succ0$. Secondly, by resorting to the Schur complement, it holds that
	\begin{align*}
		P-(\widetilde{F} P+\widetilde{G} H)^TP^{-1}(\widetilde{F} P+\widetilde{G} H)\succ0\,,
	\end{align*}
	Since $H=JP$, then we can write
	\begin{align}
		P(\widetilde{F}+\widetilde{G}J)^TP^{-1}(\widetilde{F}+\widetilde{G}J)P-PP^{-1}P&\prec0\,, \notag\\
		P((\widetilde{F}+\widetilde{G}J)^TP^{-1}(\widetilde{F}+\widetilde{G}J)-P^{-1})P&\prec0\,, \notag\\
		\widetilde{A}^TP^{-1}\widetilde{A}-P^{-1}&\prec0\,, \label{lmi2}
	\end{align}
	where $\widetilde{A}=\widetilde{F}+\widetilde{G}J$. Note that $P^{-1}=P^{-T}\succ0$ is a matrix with the same structure of $P$. From \eqref{lmi2}, the assumptions of Theorem~\ref{th2} hold, concluding the proof.
\end{proof}
\begin{proposition}
	\label{lmicl2}
	Let us consider a system with equations in the class \eqref{nonlinclassall}, where $A=F+JG$. Let $\widetilde{F}=WF$, where $W=\text{diag}(L_{p1},\dots,L_{pn})$ is the diagonal matrix defined in Section \ref{sec:th2} and $L_{pi}>0$ $\forall i$. If $\exists P=P^T$, having the structure required by Theorem \ref{th2}, and $\exists H$ such that
	\begin{align}
		\label{lmi1b}
		\begin{bmatrix}
			P-\widetilde{F}^TP\widetilde{F}-G^TH^T\widetilde{F}-\widetilde{F}^THG & G^TH^T \\
			HG & P
		\end{bmatrix}\succ0\,,
	\end{align}
	then, if we set $J=W^{-1}P^{-1}H$, the system \eqref{nonlinclassall} is $\delta$ISS.
\end{proposition}
\begin{proof}
	Firstly, note that from \eqref{lmi1b} it follows that $P\succ0$. Secondly, by resorting to the Schur complement, it holds that
	\begin{align*}
		P-\widetilde{F}^TP\widetilde{F}-G^TH^T\widetilde{F}-\widetilde{F}^THG-G^TH^TP^{-1}HG\succ0\,,
	\end{align*}
	Since $H=P\widetilde{J}$, where $\widetilde{J}=WJ$, then we can write
	\begin{align}
		\widetilde{F}^TP\widetilde{F}+G^T\widetilde{J}^TP\widetilde{F}+\widetilde{F}^TP\widetilde{J}G+G^T\widetilde{J}^TP\widetilde{J}G-P&\prec0\,, \notag\\
		\widetilde{A}^TP\widetilde{A}-P&\prec0\,, \label{lmi2b}
	\end{align}
	where $\widetilde{A}=\widetilde{F}+\widetilde{J}G$. From \eqref{lmi2b}, the assumptions of Theorem~\ref{th2} hold, concluding the proof.
\end{proof}
Now, we will show some examples of control design problems which can be solved using the results in Propositions~\ref{lmicl} and~\ref{lmicl2}. Note that, if $J$ is a full matrix whose elements are the controller parameters, then $H$ is a full matrix as well. However, if $J$ is a block matrix containing some zero blocks, then further constraints on the structure of the matrices $H$ and $P$ must be considered, as we will see in some of the following examples.\medskip\\
\textbf{Example 1. Static linear state-feedback controller}

We consider the problem of designing a state-feedback gain matrix $K$ such that the control system in Figure \ref{fig: SFCS} enjoys $\delta$ISS. The equations of the system are defined as in \eqref{sisall}. In this example the equation of the controller is
\begin{equation}
		\label{contsf}
		u_s(k)=Kx_s(k)+u_0(k)\,,
\end{equation}
where $K\in\mathbb{R}^{{m_s} \times n_s}$, $u_0(k)\in\mathbb{R}^{m_s}$ is a suitable feedforward term possibly depending upon the reference signal, and the state $x_s(k)$ is measurable or can be estimated by a suitable observer. Note that the state is certainly known in case it depends only on current and past input and output samples, e.g., a shallow NNARX where $W_0$ and $b_0$ are a priori selected and $W_0$ is invertible. In \cite{armenio2019model}, in case of ESNs, a possible observer is proposed. More generally, some insights about the design of suitable observers for generic systems in the class~\eqref{nonlinclassall} are provided in Section \ref{obsd}.

\begin{figure}[t]
	\centering
	\includegraphics[width=0.4\columnwidth]{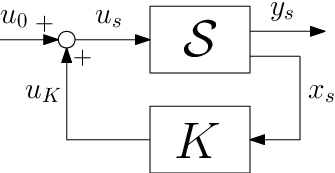}
	\caption{State-feedback control scheme: $K$ is the state-feedback gain matrix, $\mathcal{S}$ is the system to be controlled, $y_s$ is the output of the system, $x_s$ is the state of the system, $u_0$ is the feedforward term, $u_K=Kx_s$, and $u_s$ is the manipulated variable.}
	\label{fig: SFCS}
\end{figure}

Hence, the closed-loop system has equations in the class~\eqref{nonlinclassall}, where $A=A_s+B_sK$ has the structure required by Proposition~\ref{lmicl} with $F=A_s$, $G=B_s$, and $J=K$.\medskip\\
\newpage
\textbf{Example 2. Echo state dynamic output-feedback controller}

We consider the problem of designing an output-feedback controller for the control scheme in Figure \ref{fig: CS}, where the system lies in the class \eqref{sisall} and the controller is described by an ESN. We define the equations of the controller as
\begin{subequations}
	\label{esncall}
	\begin{align}
		\label{esncx}
		x_c(k+1) & = \zeta_c(W_{x_c}x_c(k)+W_{e}e(k)+W_{y_c}y_c(k)) \,,\\
		\label{esncy}
		y_c(k) & = W_{out_{1_c}}x_c(k)\,,
	\end{align}
\end{subequations}
where the direct dependence of the input in the output equation is omitted, i.e., $W_{out_{2_c}}=0_{l_c,m_c}$. By recalling that $u_s(k)=y_c(k)$ and $e(k) = r(k)-y_s(k)$, the equations of the overall closed-loop system are in the class \eqref{nonlinclassall}, where $x(k)=\begin{bmatrix}x_c(k)^T&x_s(k)^T\end{bmatrix}^T$, $f(\cdot)=\begin{bmatrix}\zeta_{c}(\cdot)^T&f_{s}(\cdot)^T\end{bmatrix}^T$, $C=\begin{bmatrix}0_{l_s,n_c}&C_s\end{bmatrix}$, $D=0_{l_s,l_s}$,  	
\begin{align*}
	A =\left[\begin{matrix}
		W_{x_c}+W_{y_c}W_{out_{1_c}}&-W_{e}C_s\\
		B_sW_{out_{1_c}}&A_s
	\end{matrix}\right]\,, \ \ B =\left[\begin{matrix}
		W_e\\
		0_{n_s,l_s}
	\end{matrix}\right]\,,
\end{align*}
$y(k)=y_s(k)$ is the overall output, and $u(k)=r(k)$ is the overall input. Since the system matrices are known and the matrices of the controller in the state equation are randomly generated, the only unknown matrix is $W_{out_{1_c}}$. Hence, it is possible to write $A=F+GJ$, where $J=\begin{bmatrix}W_{out_{1_c}}&0_{l_c,n_s}\end{bmatrix}$,
\begin{align*}
	F =\left[\begin{matrix}
		W_{x_c}&-W_{e}C_s\\
		0_{n_s,n_c}&A_s
	\end{matrix}\right]\,, \ \ G =\left[\begin{matrix}
		W_{y_c}\\
		B_s
	\end{matrix}\right]\,.
\end{align*} 
Therefore, it is possible to apply the result in Proposition \ref{lmicl}. However, in order to obtain a $J$ with the required structure, it is necessary to further constrain (i) the structure of the matrix $H$, i.e., $H=\begin{bmatrix}\widetilde{H}&0_{l_c,n_s}\end{bmatrix}$, where $\widetilde{H}\in\mathbb{R}^{l_c\times n_c}$ is a free variable, and (ii) the structure of the matrix $P$, i.e., $P=P^T=\text{diag}(P_1,P_2)$, where, in turn, $P_1\in\mathbb{R}^{n_c\times n_c}$ is diagonal, and $P_2\in\mathbb{R}^{n_s\times n_s}$ has the structure required by Theorem \ref{th2}.\medskip\\
\textbf{Example 3. Shallow NNARX dynamic output-feedback controller}

We consider the problem of designing an output-feedback controller for the control scheme in Figure \ref{fig: CS}, where the system lies in the class \eqref{sisall} and the controller is described by a shallow NNARX \eqref{nnarx}. We use the subscript $c$ to denote the matrices and dimensions of the controller. For simplicity, we set $b_c=0_{\nu_c,1}$ and $b_{0_c}=0_{l_c,1}$, which is reasonable in case normalized data are considered. We also assume that $W_{0_c}$ is a priori selected, whereas $W_{\phi_c}=\begin{bmatrix}W_{\phi_{1_c}}&W_{\phi_{2_c}}&W_{\phi_{3_c}}\end{bmatrix}$ and $W_{u_c}$ are the controller unknown matrices to be tuned. Note that the controller can be written in the state-space representation \eqref{contall}, where $x_c(k)$ is defined as in \eqref{statennarx}, $A_c$ as in \eqref{Annarx}, $B_c=\left[\begin{matrix}0_{\tau_c,m_c}^T&0_{l_c,m_c}^T&I_{m_c}&W_{u_c}^T\end{matrix}\right]^T$, $C_c=\begin{bmatrix}0_{l_c,n_c-\nu_c}&W_{0_c}\end{bmatrix}$, $D_c=0_{l_c, m_c}$, $f_c(\cdot)=\begin{bmatrix}id_{n_c-\nu_c}(\cdot)^T&\zeta_c(\cdot)^T\end{bmatrix}^T$, and $u_c(k)=e(k)=r(k)-y_s(k)$. By recalling that \mbox{$u_s(k)=y_c(k)$}, the equations of the overall closed-loop system are in the class~\eqref{nonlinclassall}, where $x(k)=\begin{bmatrix}x_s(k)^T&x_c(k)^T\end{bmatrix}^T$, \mbox{$f(\cdot)=\begin{bmatrix}f_{s}(\cdot)^T&f_{c}(\cdot)^T\end{bmatrix}^T$}, $C=\begin{bmatrix}C_s&0_{l_s,n_c}\end{bmatrix}$, $D=0_{l_s,l_s}$,
\begin{align*}
	A =\left[\begin{matrix}
		A_s&B_sC_c\\
		-B_cC_s&A_c
	\end{matrix}\right]\,, \ \ B =\left[\begin{matrix}
		0_{n_s,l_s}\\
		B_c
	\end{matrix}\right]\,,
\end{align*}
$y(k)=y_s(k)$ is the overall output, and $u(k)=r(k)$ is the overall input. Hence, it is possible to write $A=F+JG$, where
\begin{align*}
	F =\begin{bmatrix}
		A_s&B_sC_c\\
		\begin{bmatrix}0_{\tau_c,n_s}\\0_{l_c,n_s}\\-C_s\\0_{\nu_c,n_s}\end{bmatrix}&\begin{bmatrix}0_{\tau_c,l_c+m_c}&I_{\tau_c}&0_{\tau_c,\nu_c}\\0_{l_c,l_c+m_c}&0_{l_c,\tau_c}&W_{0_c}\\0_{l_s,l_c+m_c}&0_{l_s,\tau_c}&0_{l_s,\nu_c}\\0_{\nu_c,l_c+m_c}&0_{\nu_c,\tau_c}&0_{\nu_c,\nu_c}\end{bmatrix}
	\end{bmatrix}\,,\ \ \ G =\begin{bmatrix}
		-C_s&0_{l_s,n_c}\\0_{n_{j_2}-l_s,n_s}&\begin{bmatrix}
			I_{n_c-\nu_c}&0_{n_c-\nu_c,\nu_c}\\0_{l_c,n_c-\nu_c}&W_{0_c}
		\end{bmatrix}
	\end{bmatrix}\,,\ \ \ J=\begin{bmatrix}0_{n_{j_1},n_{j_2}}\\
	\begin{bmatrix}W_{u_c}&W_{\phi_c}\end{bmatrix}\end{bmatrix}\,,
\end{align*} 
$n_{j_1}=n_s+n_c-\nu_c$, and $n_{j_2}=(m_c+l_c)N_c+m_c$.
Therefore, it is possible to apply the result in Proposition \ref{lmicl2}. However, in order to obtain a $J$ with the required structure, it is necessary to further constrain (i) the structure of the matrix $H$, i.e., $H=\begin{bmatrix}0_{n_{j_1},n_{j_2}}^T&\widetilde{H}^T\end{bmatrix}^T$, where $\widetilde{H}\in\mathbb{R}^{\nu_c\times n_{j_2}}$ is a free variable, and (ii) the structure of the matrix $P$, i.e., $P=P^T=\text{diag}(P_1,P_2)$, where, in turn, $P_1\in\mathbb{R}^{n_{j_1}\times n_{j_1}}$ and $P_2\in\mathbb{R}^{\nu_c\times \nu_c}$ are  matrices with the structure required by Theorem \ref{th2}.

\subsection{Observer design}
\label{obsd}
In general, if the state is not measurable, the application of state-feedback control schemes (e.g., the one in Figure \ref{fig: SFCS}) requires the availability of a state estimate. For a system in class \eqref{nonlinclassall}, the following observer is proposed to provide a reliable estimate $\hat{x}$ of the state $x$, based on the input-output measures $u$ and $y$:
\begin{subequations}
	\label{sisallobs}
	\begin{align}
		\label{sisxobs}
		\hat{x}(k+1)&=f(A\hat{x}(k)+Bu(k)+L(y(k)-\hat{y}(k))) \\
		\label{sisyobs}
		\hat{y}(k)&=C\hat{x}(k)+Du(k)
	\end{align}
\end{subequations}
where $L$ is the observer gain to be designed according to the following result.
\begin{proposition}
	\label{lmiobs}
	Let us consider a system with equations in the class \eqref{nonlinclassall}. Let us define the diagonal matrix $$W=\text{diag}(L_{p1},\dots,L_{pn})$$ as specified in Section \ref{sec:th2}. If the observer \eqref{sisallobs} is employed, with $L$ such that
	\begin{align}
		\label{lmi1obs}
		\begin{bmatrix}
			\frac{1}{\lambda_{max}(W)^2}I_n & (A-LC)^T \\
			A-LC & I_n
		\end{bmatrix}\succ0\,,
	\end{align}
	then $\hat{x}(k)\to x(k)$ as $k\to+\infty$.
\end{proposition}
\begin{proof}
	Firstly, note that the dynamics of the estimation error is defined as $e(k)=x(k)-\hat{x}(k)$. By jointly considering \eqref{nonlinclassall} and \eqref{sisallobs}, the 2-norm of the estimation error at time instant $k+1$ can be written as 
	\begin{align*}
		\|e(k+1)\|&=\|f(Ax(k)+Bu(k))-f(A\hat{x}(k)+Bu(k)+L(Cx(k)+Du(k)-(C\hat{x}(k)+Du(k)))\,\|.
	\end{align*}
	According to Assumption \ref{Ass1} and to the definition of $W$, we can write 	
	\begin{align*}
		\|e(k+1)\|&\leq\lambda_{max}(W)\|Ax(k)+Bu(k)-(A\hat{x}(k)+Bu(k)+L(Cx(k)-C\hat{x}(k))\|\,=\\&=\,\lambda_{max}(W)\|(A-LC)(x(k)-\hat{x}(k))\|\,\leq \,\lambda_{max}(W)\|A-LC\|\|e(k)\|\,.
	\end{align*}
	Thus, the condition 
	\begin{equation}
		\label{condobs}
		\lambda_{max}(W)\|A-LC\|<1\,,
	\end{equation}
	guarantees that the estimation error converges to $0$, i.e., $\|e(k)\|\to0$ and $\hat{x}(k)\to x(k)$ as $k\to+\infty$. Note that \eqref{condobs} is equivalent to the following condition $$\frac{1}{\lambda_{max}(W)^2}I_n-(A-LC)^TI_n(A-LC)\succ0\,,$$ which can be recast as \eqref{lmi1obs} in view of the Schur complement.
\end{proof}
The study of the convergence rate of the observer as well as the analysis of the case in which a non-measurable disturbance acts on the system state and/or output will be matter of future research. 

\subsection{Control schemes with zero steady-state error}

\begin{figure}[t]
	\centering
	\subfloat[]{\includegraphics[width=0.55\columnwidth]{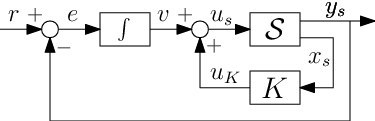}}
	\\[0.3cm]
	\subfloat[]{\includegraphics[width=0.55\columnwidth]{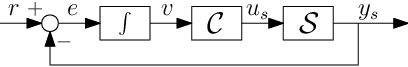}}
	\\[0.3cm]
	\subfloat[]{\includegraphics[width=0.55\columnwidth]{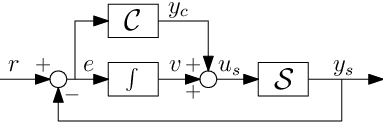}}
	\caption{Closed-loop control schemes with explicit integral action, where $\int$ is the discrete-time integrator, and $\mathcal{S}$ the system to be controlled. (a) Static state-feedback controller with gain $K$. (b) Integrator in series to the controller~$\mathcal{C}$. (c) Integrator in parallel to the controller $\mathcal{C}$.}
	\label{fig: CSINT}
\end{figure}


In this section the possibility to guarantee zero steady-state error in case of tracking of piecewise constant reference signals for systems in the class \eqref{nonlinclassall} is investigated. This can be guaranteed, e.g., using the control schemes in Figure \ref{fig: CSINT}, where the system $\mathcal{S}$ is in the class \eqref{sisall}, the controller $\mathcal{C}$ is in the class \eqref{contall}, and the block ``$\int$'' denotes a discrete-time integrator with equation
\begin{subequations}
	\label{intall}
	\begin{align}
		\label{intx}
		\eta(k+1) & = \eta(k)+e(k)\,, \\
		\label{inty}
		v(k) & = M(\eta(k)+e(k))\,,
	\end{align}
\end{subequations}
where $\eta(k)\in\mathbb{R}^{l_s}$ is the state of the integrator, and \mbox{$M\in\mathbb{R}^{l_s\times l_s}$} is its gain matrix. In the following proposition we will prove that all the control schemes depicted in Figure~\ref{fig: CSINT} lead to a common general model of type
\begin{subequations}
	\label{sisint}
	\begin{align}
		\label{sisintx1}
		\chi(k+1)&=f_\chi(A_\chi \chi(k)+A_\eta\eta(k)+B_\chi r(k))\,, \\
		\label{sisintx2}			
		\eta(k+1) & = -C_\chi \chi(k)+\eta(k)+r(k)\,, \\
		\label{sisinty}
		y_s(k) & = C_\chi \chi(k)\,,
	\end{align}
\end{subequations}
where $r(k)\in\mathbb{R}^{l_s}$ is the reference input, $y_s(k)\in\mathbb{R}^{l_s}$ is the output of the system, $\chi(k)\in\mathbb{R}^{n_\chi}$ is a vector of states, \mbox{$f_\chi(\cdot)\in\mathbb{R}^{n_\chi}$} is a vector of scalar functions, $A_\chi\in\mathbb{R}^{n_\chi\times n_\chi}$, $A_\eta\in\mathbb{R}^{n_\chi\times l_s}$, $B_\chi\in\mathbb{R}^{n_\chi\times l_s}$, and $C_\chi\in\mathbb{R}^{l_s\times n_\chi}$. We can state the following result.
\begin{proposition}
	\label{int1}
	Let the controller $\mathcal{C}$ be in the class \eqref{contall} and the system $\mathcal{S}$ be in the class \eqref{sisall}. The equations of the closed-loop systems in Figure \ref{fig: CSINT} lie in the class \eqref{sisint}. Moreover, the set of systems in the class \eqref{sisint} is a subset of the set of systems in the class \eqref{nonlinclassall}.
\end{proposition}
\begin{proof}
	From Figure \ref{fig: CSINT}(a), by jointly considering \eqref{sisall}, \eqref{intall}, $e(k)=r(k)-C_sx_s(k)$, and $u_s(k)=Kx_s(k)+v(k)$, we obtain that the equations of the control system are in the class \eqref{sisint}, where $\chi(k)=x_s(k)$, $f_\chi(\cdot)=f_s(\cdot)$, $A_\chi=A_s+B_sK-B_sMC_s$, $A_\eta=B_sM$, $B_\chi=A_\eta$, and $C_\chi=C_s$.
	
	From Figure \ref{fig: CSINT}(b), by jointly considering \eqref{contall}, \eqref{sisall}, \eqref{intall}, $e(k)=r(k)-C_sx_s(k)$, $u_c(k)=v(k)$, and $u_s(k)=y_c(k)$, we obtain that the equations of the control system are in the class \eqref{sisint}, where $\chi(k)=\begin{bmatrix}x_c(k)^T&x_s(k)^T\end{bmatrix}^T$, \mbox{$f_\chi(\cdot)=\begin{bmatrix}f_{c}(\cdot)^T&f_{s}(\cdot)^T\end{bmatrix}^T$}, $B_\chi=A_\eta$, $C_\chi=\begin{bmatrix}0_{l_s,n_c}&C_s\end{bmatrix}$,
	\begin{align*}
		A_\chi =\left[\begin{matrix}
			A_c&- B_c MC_s\\
			B_s C_c&A_s- B_s D_c MC_s
		\end{matrix}\right]\,,\ \text{and}\ A_\eta =\left[\begin{matrix}
		 B_c M\\
		 B_s D_c M
	\end{matrix}\right]\,.
	\end{align*} 
	
	From Figure \ref{fig: CSINT}(c), by jointly considering \eqref{contall}, \eqref{sisall}, \eqref{intall}, $e(k)=r(k)-C_sx_s(k)$, $u_c(k)=e(k)$, and $u_s(k)=v(k)+y_c(k)$, we obtain that the equations of the control system are in the class \eqref{sisint}, where $\chi(k)=\begin{bmatrix}x_c(k)^T&x_s(k)^T\end{bmatrix}^T$, 	$f_\chi(\cdot)=\begin{bmatrix}f_{c}(\cdot)^T&f_{s}(\cdot)^T\end{bmatrix}^T$, $C_\chi=\begin{bmatrix}0_{l_s,n_c}&C_s\end{bmatrix}$,
	\begin{align*}
		A_\chi =\left[\begin{matrix}
			A_c&-B_c C_s\\
			B_s C_c&A_s- B_s MC_s-B_s D_c C_s
		\end{matrix}\right]\,,\ A_\eta =\left[\begin{matrix}
			0_{n_c,l_s}\\
			B_s M
		\end{matrix}\right],
	\end{align*} 
	and $B_\chi=\begin{bmatrix}B_c^T&(B_sM+B_sD_c)^T\end{bmatrix}^T$.
	
	Finally, we can easily see that all the systems \eqref{sisint} are in the class \eqref{nonlinclassall}, where $x(k)=\begin{bmatrix}\chi(k)^T&\eta(k)^T\end{bmatrix}^T$, $y(k)=y_s(k)$, $u(k)=r(k)$, $f(\cdot)=\begin{bmatrix}f_\chi(\cdot)^T&id_{l_s}(\cdot)^T\end{bmatrix}^T$, $C=\begin{bmatrix}C_\chi&0_{l_s,l_s}\end{bmatrix}$,
	\begin{align*}
		A =\left[\begin{matrix}
			A_\chi&A_\eta\\
			-C_\chi&I_{l_s}
		\end{matrix}\right]\,,\ B =\left[\begin{matrix}
			B_\chi\\
			I_{l_s}
		\end{matrix}\right]\,,
	\end{align*} 
	and $D=0_{l_s,l_s}$. This concludes the proof.
\end{proof}

Given an equilibrium point, provided that the control schemes in Figure~\ref{fig: CSINT} are $\delta$ISS, the zero steady-state error is ensured by the explicit integral action, since such an equilibrium is globally asymptotically stable according to Definition~\ref{deltaISSdef}. However, there are some cases in which the $\delta$ISS property cannot be enforced to the control schemes in Figure~\ref{fig: CSINT}, e.g., if the output $y_s(k)$ of the system is bounded. This is the case of some RNN architectures where all the activation functions are bounded, e.g., the hyperbolic tangent or the sigmoid function. Some examples are ESNs where $W_{out_{2}}=0_{l_s,m_s}$ and $\zeta(\cdot)=\tanh(\cdot)$ \cite{d2022recurrent} or shallow NNARXs where $\zeta(\cdot)=\tanh(\cdot)$. In this regard, the following result holds. 
\begin{proposition}
	\label{int}
	Let us consider a control system with equations in the class \eqref{sisint}. Let at least one output of the system be bounded, i.e., $y_{si}(k)\geq y_{min}$ and/or $y_{si}(k)\leq y_{max}$ for all $k$, and for at least one $i=1,...,l_s$, where $y_{max},y_{min}\in\mathbb{R}$. Then, the control system \eqref{sisint} cannot enjoy the $\delta$ISS property. 
\end{proposition}
\begin{proof}
	Recall that $\delta$ISS implies ISS (see \cite{bayer2013discrete}), and that ISS implies the bounded-input bounded-state property (see \cite{sontag1996new}). In case $y_{si}(k)\leq y_{max}$ for all $k$, if we take a constant \mbox{$r_i(k)=\bar{r}>y_{max}$} for all $k$, we have from \eqref{sisintx2} that the state $\eta_i(k)\to+\infty$ for $k\to+\infty$, since $y_{si}(k)=C_{\chi_i} \chi(k)\leq y_{max}$ for all $k$, where $C_{\chi_i}$ is the $i$-th row of $C_{\chi}$. It follows that the bounded-input bounded-state property does not hold, which is necessary for $\delta$ISS. The same result holds for the case in which $y_{si}(k)\geq y_{min}$ $\forall k$, since $\eta_i(k)\to-\infty$ for $k\to+\infty$ by setting $r_i(k)\!=\!\bar{r}\!<\!y_{min}$ $\forall k$. This concludes the proof.
\end{proof}
Hence, if we consider the system \eqref{sisint}, where $f_\chi(\cdot)$ is composed of bounded nonlinear globally Lipschitz continuous functions, we have that the assumptions of Theorem~\ref{th2} can never be fulfilled, as a consequence of the result in Proposition~\ref{int}. Nevertheless, if we consider a system structure where at least a state equation in \eqref{sisintx1} is linear, and the latter state directly affects the output in \eqref{sisinty}, then $y_s(k)$ may be unbounded, and the condition in Theorem~\ref{th2}, as well as $\delta$ISS, can be enforced. The following example corroborates the statement in Proposition~\ref{int} and our remark.\medskip\\
\textbf{Example: state-feedback control with explicit integrator}

We consider a SISO system in the class \eqref{sisall} with two states, where $f_s(\cdot)=\begin{bmatrix}\tanh(\cdot)&f_2(\cdot)\end{bmatrix}^T$. Let us consider a state-feedback control law and an explicit integral action \eqref{intall} as in Figure \ref{fig: CSINT}(a), i.e., $u_s(k)=Kx_s(k)+M(\eta(k)+e(k))$, where $e(k)=r(k)-C_sx_s(k)$. As stated in Proposition \ref{int1}, the closed-loop system is in the class \eqref{sisint}, and by extension also in the class \eqref{nonlinclassall}. The matrix $A$ of the closed-loop system~is
\begin{align*}
	A =\left[\begin{matrix}
		A_s+B_sK- B_sMC_s& B_sM\\
		-C_s&1
	\end{matrix}\right]\,,
\end{align*} 
and can be rewritten as $A=F+GJ$, where
\begin{align*}
	F =\left[\begin{matrix}
		A_s&0_{2,1}\\
		-C_s&1
	\end{matrix}\right]\,, \ \ G =\left[\begin{matrix}
		B_s\\
		0
	\end{matrix}\right]\,,
\end{align*} 
and $J=\begin{bmatrix}K-M C_s&M\end{bmatrix}$ takes the role of the control gain. Note that $\begin{bmatrix}K&M\end{bmatrix}=JE^{-1}$, where $E=\begin{bmatrix}
	I_{2}&0_{2,1}\\-C_s&1
\end{bmatrix}$. 

Let us choose $A_s\!=\!\begin{bmatrix}
	-0.4686&1.0984\\0&1.15
\end{bmatrix}$, \mbox{$B_s\!=\!\begin{bmatrix}
	0.7015\\-2.0518
\end{bmatrix}$}, $C_s=\begin{bmatrix}
	-0.3538&-0.8236
\end{bmatrix}$, and $f_2(\cdot)=id(\cdot)$. \vspace{1mm}

According to Theorem \ref{th2}, matrix $P$ must have the following structure: $P=\text{diag}(p_1,P_2)$, where $p_1\in\mathbb{R}$ and $P_2\in\mathbb{R}^{2\times 2}$. A feasible solution is returned by solving \eqref{lmi1} (where $W=I_3$) with YALMIP and MOSEK~\cite{lofberg2004yalmip,mosek}, ensuring that the closed-loop system enjoys the $\delta$ISS property according to Proposition \ref{lmicl}. The following matrices and parameters are obtained: $P=\begin{bmatrix}
	2.1317&0&0\\0&0.588&-0.5283\\0&-0.5283&1.3569
\end{bmatrix}$, $H=\begin{bmatrix}
	0.1694&0.2875&-0.1088
\end{bmatrix}$, $K=\begin{bmatrix}
	0.0196&0.5016
\end{bmatrix}$, and $M=0.1694$. 

\begin{figure}[t]
	\centering
	\includegraphics[width=0.7\columnwidth]{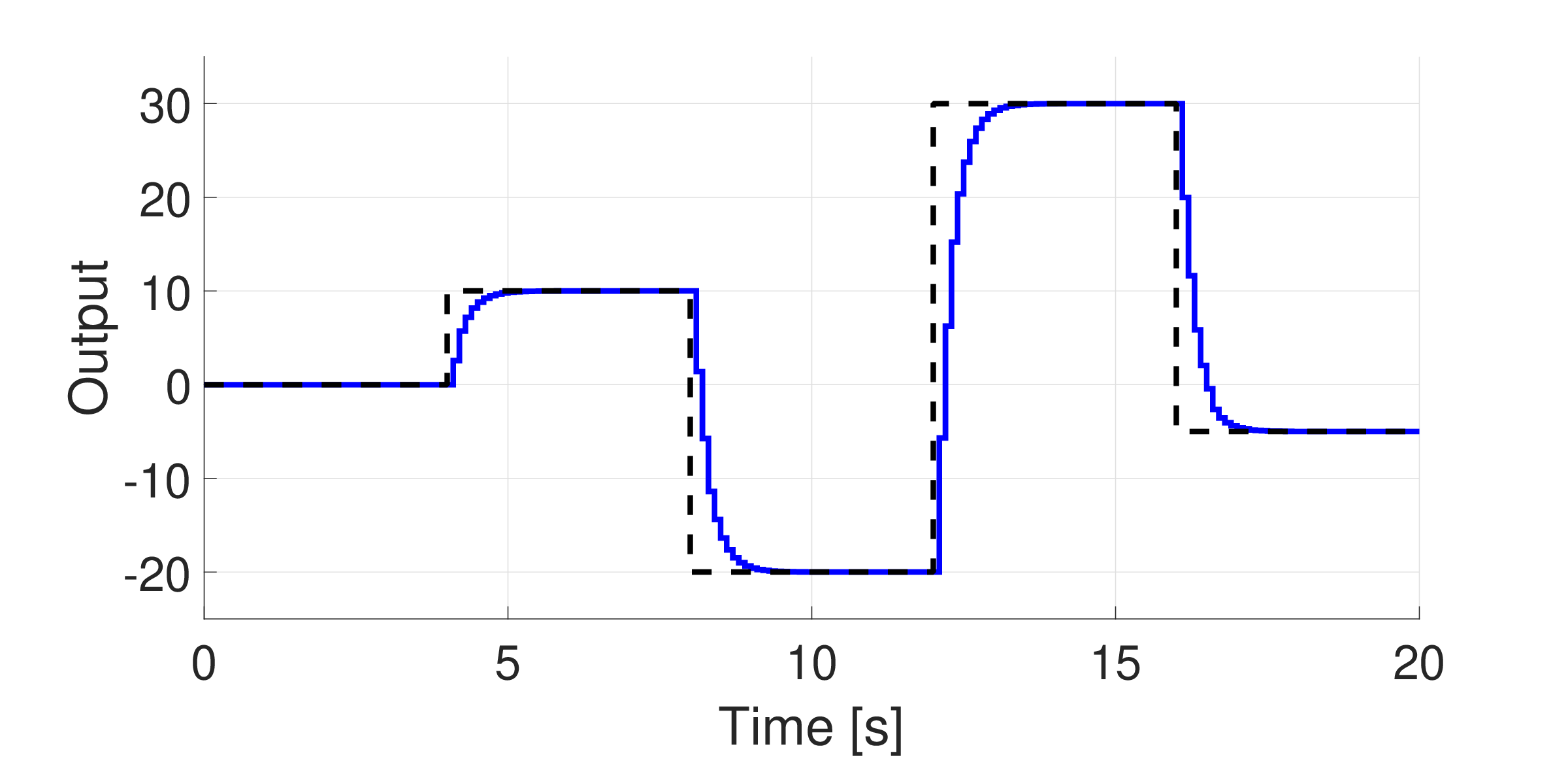}
	\caption{Output trajectory of the closed-loop discrete-time system with state-feedback and explicit integral action. Black dashed line: reference signal trajectory; blue line: output signal trajectory.}
	\label{outint}
\end{figure}
\medskip
In Figure \ref{outint} the reference tracking results of the closed-loop system are depicted, where we can see that the closed-loop system enjoys $\delta$ISS and its equilibria are asymptotically stable even if the open-loop system displays unstable dynamics, as we can see from the second (linear) state equation of the open-loop system. Furthermore, due to the explicit integral action, a zero steady-state error is also achieved. 

On the other hand, if we take $f_2(\cdot)=\tanh(\cdot)$, matrix $P$ is constrained to have the following structure, on the basis of Theorem \ref{th2}: $P=\text{diag}(p_1,p_2,p_3)$, where $p_1,p_2,p_3\in\mathbb{R}$. An unfeasible solution is returned by solving \eqref{lmi1} (where \mbox{$W=I_3$}), corroborating the statement in Proposition \ref{int}. 

If the control schemes in Figure \ref{fig: CSINT} cannot be used to achieve zero steady-state error while ensuring $\delta$ISS for the closed-loop system, a possible solution to improve the static performance could be the use of a $\delta$ISS feedforward compensator in the class \eqref{nonlinclassall}. Since it is dynamic, this compensator can be used to enhance both static and dynamic performances. Moreover, provided that the closed-loop system is $\delta$ISS, the addition of a $\delta$ISS feedforward compensator preserves the $\delta$ISS of the overall control system. This fact follows straightforwardly from the result in Proposition \ref{cas}, since such a component is placed in series to the closed-loop system. Future research will address methods for the design of feedforward compensators, as well as the analysis of the dynamic performances of the control system. 

\section{Simulation results}
\label{sec:sim}
In order to validate the theoretical results in the previous sections we propose here a simulation example. The case study consists of the control of the following ESN-based nonlinear SISO system with $n_s=8$ states
\begin{subequations}
	\label{esnsim}
	\begin{align}
		x_s(k+1) & \!=\! \zeta_s(W_{x_s}x_s(k)+W_{u_s}u_s(k)+W_{y_s}y_s(k)) \,,
		\label{net_statesim}\\
		y_s(k) & \!=\! W_{out_{1_s}}x_s(k)\,,
		\label{net_outputsim}
	\end{align}
\end{subequations}
where the direct dependence of the input in the output equation is absent, i.e., $W_{out_{2_s}}=0$, and $\zeta_{s_i}(\cdot)=id(\cdot)$ for $i=1,\dots,5$, whereas $\zeta_{s_i}(\cdot)=\tanh(\cdot)$ for $i=6,7,8$. The structure of the state equations, both nonlinear and linear, is such that the output can take unbounded values, paving the way to the possible inclusion of an explicit integral action (according to Proposition \ref{int}), and to more freedom in the control design (e.g., see Theorem \ref{th2}). 

The model \eqref{esnsim} is identified using a noiseless dataset containing $70000$ normalized input-output data collected with a sampling time $T_s=25\ s$ from a simulated \textit{pH} neutralization process (see \cite{armenio2019model, armenio2019echo, hall1989modelling} for a detailed description of the system model). The \textit{pH} process is a nonlinear SISO dynamical system where the input is the alkaline flowrate, and the output is the \textit{pH} concentration. The training input data consist of a multilevel pseudo-random signal (MPRS) \cite{armenio2019echo}, whose amplitude is in the range $[12, 16]\ mL/s$. The training is carried out according to the ``ESN training algorithm'' (see \cite{armenio2019model} for an accurate description). Basically, $W_{x_s}$, $W_{u_s}$, and $W_{y_s}$ are randomly generated, whereas $W_{out_{1_s}}$ is obtained by solving a least squares problem based on the available dataset, where the initial $500$ data points are discarded to accommodate the effect of the initial transient. To test the identification performance, the following fitting index is calculated over a validation dataset composed of $30000$ new normalized input-output data
\begin{equation}
	\label{fit}
	FIT_\% = 100\cdot\left(1-\frac{\|\vec{y}-\vec{y}_s\|}{\|\vec{y}-\bar{y}\|}\right)\in(-\infty,100]\,,
\end{equation}
where $\vec{y}$ is the real system output sequence, $\vec{y}_s$ is the output sequence obtained with \eqref{esnsim} and $\bar{y}$ is a vector with all the elements equal to the mean value of the real output sequence $\vec{y}$. A satisfactory fitting $FIT_\%=84.1452\%$ is achieved. 

\begin{figure}[t]
	\centering
	\includegraphics[width=0.6\columnwidth]{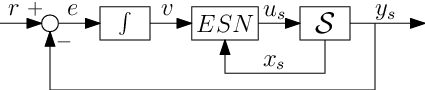}
	\caption{Control scheme with explicit integral action and ESN in series: $\int$ is the discrete-time integrator, $ESN$ the echo state dynamic controller, and $\mathcal{S}$ the system to be controlled.}
	\label{fig: SIMCS}
\end{figure}

The control objective in this example is the achievement of perfect asymptotic tracking of constant reference signals. To this aim, the control architecture in Figure \ref{fig: SIMCS} is taken into account, where an explicit integral action is also embedded. In particular, $\mathcal{S}$ is the system \eqref{esnsim} to be controlled, whose state is assumed measurable, ``$\int$'' is defined as in \eqref{intall} by setting $M=1$, and ``$ESN$'' is a non-strictly proper ESN controller with $n_c=5$ states, whose input vector contains the integrator output $v(k)$ jointly to the system state vector $x_s(k)$. Overall, the controller equations are the following
\begin{subequations}
	\label{consim}
	\begin{align}
		x_c(k+1) &= \zeta_c(W_{x_c}x_c(k)+W_{u_c}u_c(k)+W_{y_c}y_c(k))\,,
		\label{consimesn}\\
		\eta(k+1) &= \eta(k)+e(k)\,,
		\label{consimint}\\
		y_c(k) &= W_{out_{1_c}}x_c(k)+W_{out_{2_c}}u_c(k)\,,
		\label{consimout}
	\end{align}
\end{subequations}
where $\zeta_{c_i}(\cdot)=\tanh(\cdot)$ for all $i=1,\dots,5$, $u_c(k)=\begin{bmatrix}v(k)&x_s(k)^T\end{bmatrix}^T$, and $v(k)=\eta(k)+e(k)$. Let us define $W_{u_c}=\begin{bmatrix}W_{u_{cv}}&W_{u_{cx}}\end{bmatrix}$, and $W_{out_{2_c}}=\begin{bmatrix}W_{out_{2_{cv}}}&W_{out_{2_{cx}}}\end{bmatrix}$, where $W_{u_{cv}}\in\mathbb{R}^{n_c}$, $W_{u_{cx}}\in\mathbb{R}^{n_c\times n_s}$, $W_{out_{2_{cv}}}\in\mathbb{R}$, and $W_{out_{2_{cx}}}\in\mathbb{R}^{1\times n_s}$. Also, let us introduce $x(k)=\begin{bmatrix}x_c(k)^T&\eta(k)&x_s(k)^T\end{bmatrix}^T$, and $f(\cdot)=\begin{bmatrix}\zeta_{c}(\cdot)^T&id(\cdot)&\zeta_{s}(\cdot)^T\end{bmatrix}^T$. By recalling that $u_s(k)=y_c(k)$ and $e(k) = r(k)-y_s(k)$, we can write the closed-loop system equations in the class \eqref{nonlinclassall}, where the matrix $A$ is defined as
\begin{align*}
	A =\left[\begin{matrix}
		W_{x_c}+W_{y_c}W_{out_{1_c}}&W_{u_{cv}}+W_{y_c}W_{out_{2_{cv}}}&A_{13}\\
		0_{1,n_c}&1&-W_{out_{1_s}}\\
		W_{u_s}W_{out_{1_c}}&W_{u_s}W_{out_{2_{cv}}}&A_{33}
	\end{matrix}\right]\,,
\end{align*} 
where $A_{13}=W_{u_{cx}}-W_{u_{cv}}W_{out_{1_s}}+W_{y_c}(W_{out_{2_{cx}}}-W_{out_{2_{cv}}}W_{out_{1_s}})$, and $A_{33}=W_{x_s}+W_{y_s}W_{out_{1_s}}+W_{u_s}(W_{out_{2_{cx}}}-W_{out_{2_{cv}}}W_{out_{1_s}})$. Since $W_{x_c}$, $W_{u_c}$, and $W_{y_c}$ are known randomly generated matrices, we can rewrite $A$ as $F+GJ$, where
\begin{align*}
	F =\left[\begin{matrix}
		W_{x_c}&W_{u_{cv}}&W_{u_{cx}}-W_{u_{cv}}W_{out_{1_s}}\\
		0_{1,n_c}&1&-W_{out_{1_s}}\\
		0_{n_s,n_c}&0_{n_s,1}&W_{x_s}+W_{y_s}W_{out_{1_s}}		
	\end{matrix}\right]\,, \ \ G =\left[\begin{matrix}
		W_{y_c}\\
		0\\		
		W_{u_s}
	\end{matrix}\right]\,,
\end{align*} 
and $J\!=\!\begin{bmatrix}W_{out_{1_c}}&W_{out_{2_{cv}}}&W_{out_{2_{cx}}}\!-\!W_{out_{2_{cv}}}\!W_{out_{1_s}}\end{bmatrix}$ takes the role of the control gain. Note that the unknown controller parameters can be computed as $\begin{bmatrix}W_{out_{1_c}}&W_{out_{2_{c}}}\end{bmatrix}=JE^{-1}$, where $$E=\begin{bmatrix}
	I_{n_c}&0_{n_c,1}&0_{n_c,n_s}\\
	0_{1,n_c}&1&-W_{out_{1_s}}\\
	0_{n_s,n_c}&0_{n_s,1}&I_{n_s}
\end{bmatrix}\,.$$
According to Proposition \ref{lmicl}, the matrix $P$ must have the following block diagonal structure: $P=\text{diag}(P_{D_1},P_F,P_{D_2})$, where $P_{D_1}\in\mathbb{R}^{5\times 5}$ and $P_{D_2}\in\mathbb{R}^{3\times 3}$ are diagonal matrices, whereas $P_F\in\mathbb{R}^{6\times 6}$ is a full symmetric matrix. A feasible solution is returned by solving \eqref{lmi1}, where $W=I_{n_c+n_s+1}$, with YALMIP and MOSEK~\cite{lofberg2004yalmip,mosek}, ensuring that the closed-loop system enjoys the $\delta$ISS property due to Proposition \ref{lmicl}.

\begin{figure}[t]
	\centering
	\includegraphics[width=0.7\columnwidth]{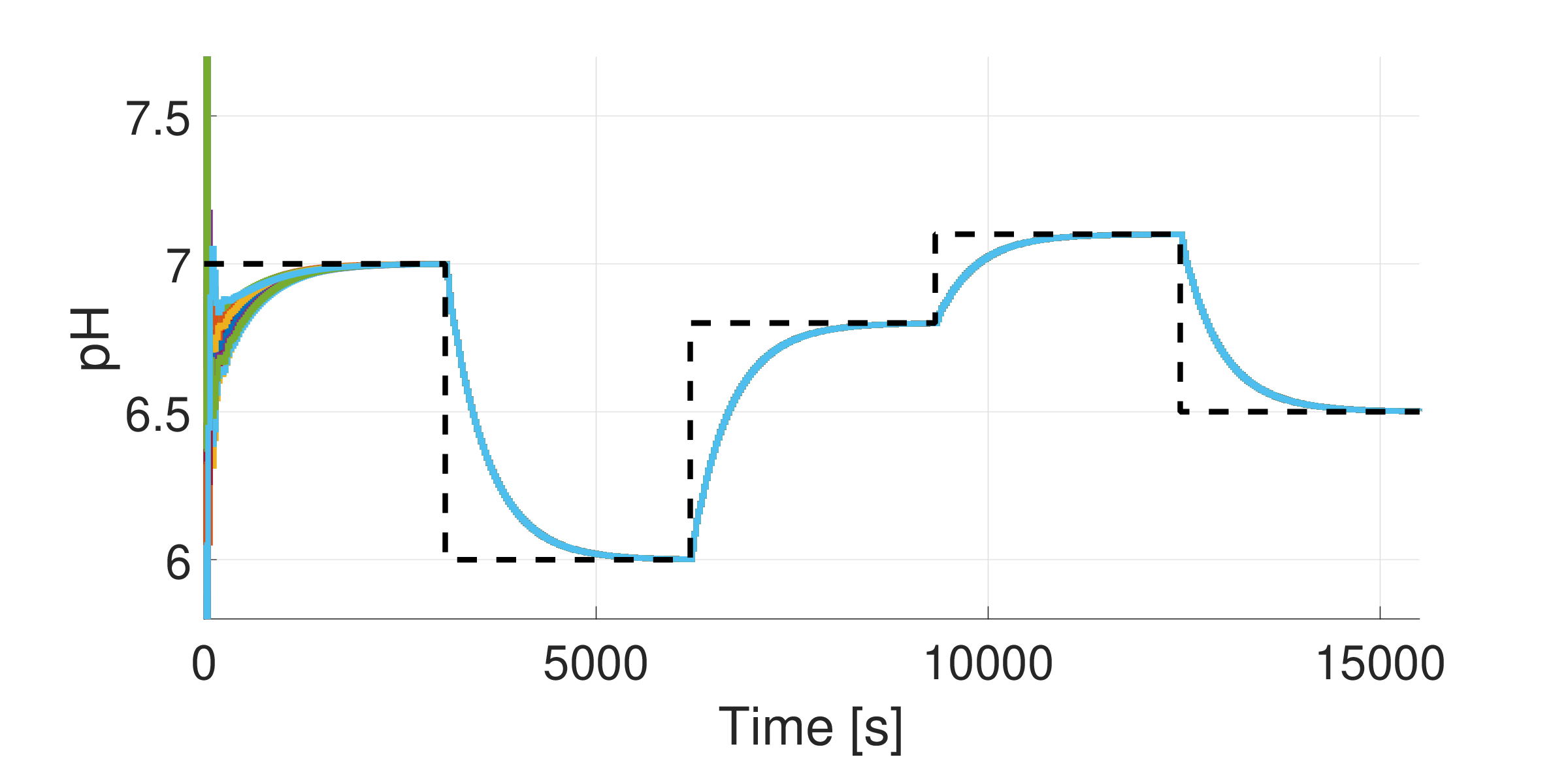}
	\caption{Denormalized output trajectories of the closed-loop discrete-time system starting from different initial conditions. Black dashed line: reference trajectory; colored lines: output trajectories for $20$ different initial conditions.}
	\label{outsim}
\end{figure} 

In Figure \ref{outsim} the reference tracking results of the closed-loop system starting from $20$ different random initial conditions are depicted, where we can see that the equilibria are asymptotically stable and the output trajectories converge to each other in view of the $\delta$ISS property. Moreover, due to the explicit integral action, a zero steady-state error is achieved. 

The previously tuned controller is also tested on the \textit{pH} process physics-based simulator. To this aim, some remarks are due. Firstly, in the control scheme, a denormalization of the control variable $u_s$ and a normalization of the output $y_s$ are performed upstream and downstream of the process, respectively. A normalization of the reference signal is also carried out. The same normalization parameters applied in the identification are employed. 

The second issue to be considered concerns the fact that the state measurements of the system \eqref{esnsim}, necessary for the state-feedback, are not available in this second case. Hence, a suitable observer is required. As suggested in Section \ref{obsd}, the following observer is tuned:
\begin{subequations}
	\label{esnallobs}
	\begin{align}
		\hat{x}_s(k+1)&=\zeta_s((W_{x_s}+W_{y_s}W_{out_{1_s}})\hat{x}_s(k)+W_{u_s}u_s(k)+L(y_s(k)-\hat{y}_s(k)))\,, \\
		\label{esnyobs}
		\hat{y}_s(k)&=W_{out_{1_s}}\hat{x}_s(k)\,,
	\end{align}
\end{subequations}
where the observer gain $L$ is designed using the result in Proposition \ref{lmiobs}. 

\begin{figure}[t]
	\centering
	\includegraphics[width=0.7\columnwidth]{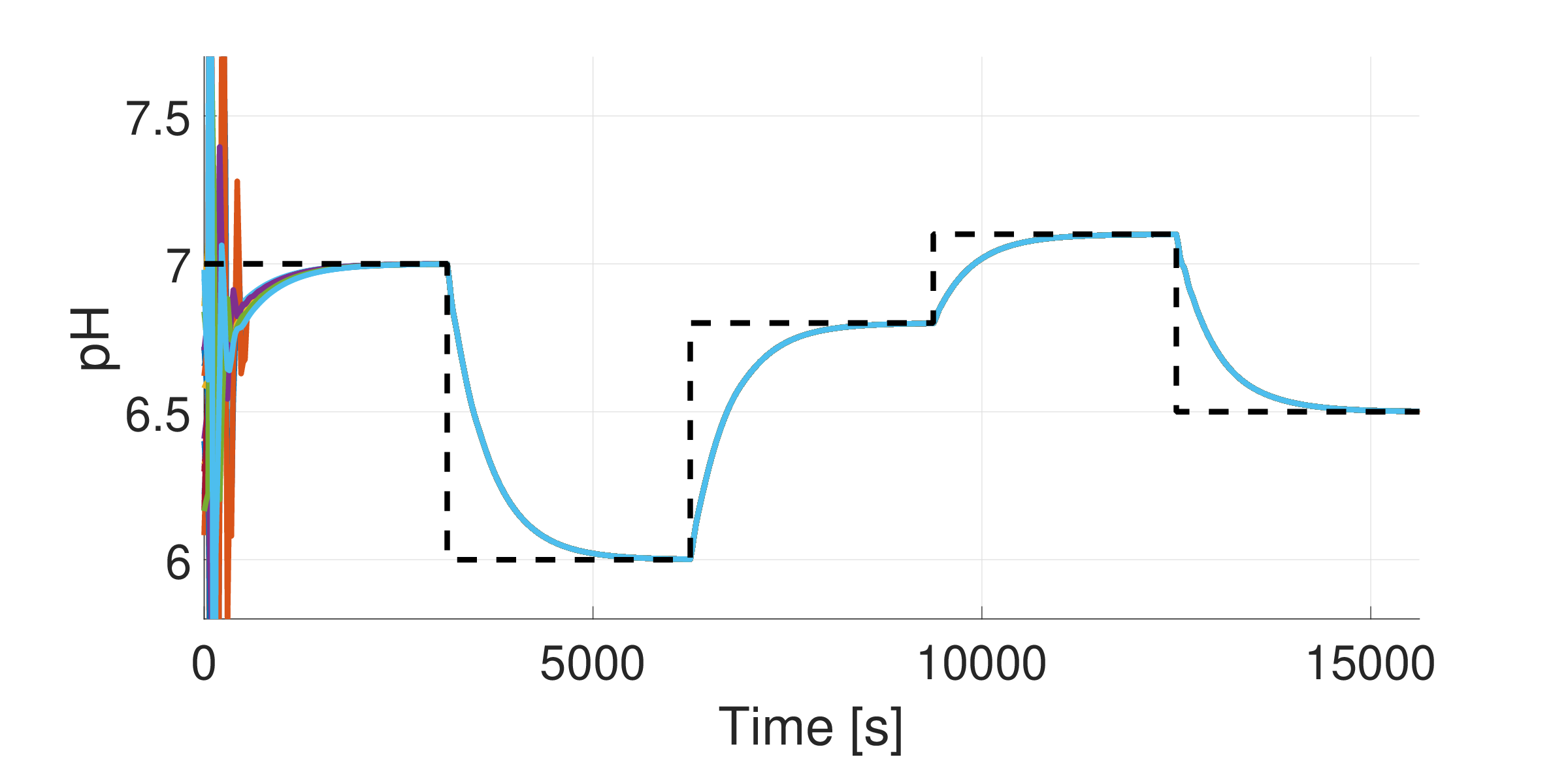}
	\caption{Output trajectories of the closed-loop system using the simulated \textit{pH} process and starting from different initial conditions. Black dashed line: reference trajectory; colored lines: output trajectories for $20$ different initial conditions.}
	\label{outpH}
\end{figure} 

In Figure \ref{outpH} the reference tracking results of the closed-loop system using the simulated \textit{pH} process and starting from $20$ different random initial conditions are represented, where we can see that the equilibria are asymptotically stable and, due to the explicit integral action, a zero steady-state error is also achieved.

\section{Conclusions}
\label{sec:conclusions}
In this paper, a $\delta$ISS condition for a generic class of RNNs has been proposed. The reduced conservativeness of this condition with respect to other conditions in the literature has been proven. Moreover, the use of this condition for control design purposes has been widely investigated. Simulation results have corroborated the effectiveness of the theoretical results.\\
Future work will first tackle a number of issues remained open in this work. Firstly, a possible extension to alternative conditions to the one in Theorem \ref{th2} in order to include other classes of RNNs (e.g., multi-layers NNARX or long short-term memory networks) will be addressed. Secondly, the development of a (possibly data-based) cost function which takes into account the desired dynamic performances of the control system will be investigated. Then, we will address the design of a suitable $\delta$ISS feedforward compensator to be used to achieve static precision in case an explicit integral action cannot be embedded in the control scheme. Furthermore, the convergence rate of the observer together with the analysis of the case in which a non-measurable disturbance acts on the system will be studied. Finally, the application of the theoretical results in this paper to an experimental apparatus will be carried out. 

\section*{Appendix}
\renewcommand*{\proofname}{Proof of Proposition \ref{Prop4}}
\begin{proof}
	Firstly, note that the condition $\|W_x^*\|<1$ is equivalent to the condition 
	\begin{align}
		\label{condid}
		{W_x^*}^TI_nW_x^*-I_n\prec0.
	\end{align} 
	Also note that, since  $L_{p}\leq1$, we can set $W=I_n$. Then, for ESNs, Theorem \ref{th2} states that system \eqref{esn} in the state-space form \eqref{nonlinclassall} is $\delta$ISS if $\exists P=P^T=\text{diag}(P_1,P_2)\succ0$, where $P_1\in\mathbb{R}^{\nu\times \nu}$ is diagonal and $P_2\in\mathbb{R}^{m\times m}$ is full, such that
	\begin{align*}
		\begin{split}
			A^TPA-P=\left[\begin{matrix}
				W_x^{*T}P_1W_x^*-P_1&W_x^{*T}P_1W_yW_{out_{2}}\\
				W_{out_{2}}^TW_y^TP_1W_x^*&W_{out_{2}}^TW_y^TP_1W_yW_{out_{2}}-P_2
			\end{matrix}\right]\prec0\,,
		\end{split}
	\end{align*}
	where $A$ is defined as in \eqref{Aesn}.\\
	We define a generic vector $v=\begin{bmatrix}v_1^T&v_2^T\end{bmatrix}^T$, where $v_1\in\mathbb{R}^{\nu}$ and $v_2\in\mathbb{R}^{m}$. Therefore, we want to prove that $v^T(A^TPA-P)v<0$ for any $v\neq0_{n,1}$. Hence, $v^T(A^TPA-P)v=v_1^T(W_x^{*T}P_1W_x^*-P_1)v_1 +2v_2^TW_{out_{2}}^TW_y^TP_1W_x^*v_1+v_2^T(W_{out_{2}}^TW_y^TP_1W_yW_{out_{2}}-P_2)v_2$. 
	Furthermore, in view of Property \ref{Prop1}, for any $\tau\neq0$,
	\begin{align*}
		v^T(A^TPA-P)v\leq v_1^T((1+\tau^2)W_x^{*T}P_1W_x^*-P_1)v_1+v_2^T((1+\frac{1}{\tau^2})W_{out_{2}}^TW_y^TP_1W_yW_{out_{2}}-P_2)v_2\,.
	\end{align*}
	Now, by setting $P_1=I_n$, it follows from \eqref{condid} that $(1+\tau^2)W_x^{*T}P_1W_x^*-P_1\prec0$ for a sufficiently small $\tau$. Moreover, note that $A_\tau^*=(1+\frac{1}{\tau^2})W_{out_{2}}^TW_y^TP_1W_yW_{out_{2}}\preceq\lambda_{max}(A_\tau^*)I_{m}$, where $\lambda_{max}(A_\tau^*)\geq0$. By setting $P_2=\lambda_pI_m$, with $\lambda_p>\lambda_{max}(A_\tau^*)$, we have that $(1+\frac{1}{\tau^2})W_{out_{2}}^TW_y^TP_1W_yW_{out_{2}}-P_2\prec0$. This implies that $v^T(A^TPA-P)v<0$ for any $v\neq0_{n,1}$, and the statement follows.
\end{proof}
\renewcommand*{\proofname}{Proof of Proposition \ref{Prop6}}
\begin{proof}
	The objective is to prove that the fulfilment of the assumption of Proposition \ref{deltaISSNNARX} implies the fulfilment of the assumptions of Theorem \ref{th2}. The latter, for a 1-layer NNARX, require the existence of a matrix $P=P^T\succ0$ such that
	$$\widetilde{A}^TP\widetilde{A}-P\prec0\,,$$
	where $\widetilde{A}=WA$, $A=\left[\begin{matrix}
		0_{\tau,l+\widetilde{m}}&I_{\tau}&0_{\tau,\nu}\\
		0_{l,l+\widetilde{m}}&0_{l,\tau}&W_0\\
		0_{\widetilde{m},l+\widetilde{m}}&0_{\widetilde{m},\tau}&0_{\widetilde{m},\nu}\\
		W_{\phi_1}&W_{\phi_2}&W_{\phi_3}W_0
	\end{matrix}\right]$, $W=\left[\begin{matrix}
			I_{n-\nu}&0_{n-\nu,\nu}\\
			0_{\nu,n-\nu}&L_pI_\nu\\	
		\end{matrix}\right],\  
		P=\left[\begin{matrix}
			\widetilde{P}&0_{n-\nu,\nu}\\
			0_{\nu,n-\nu}&P_D\\	
		\end{matrix}\right],$ $\widetilde{P}\in\mathbb{R}^{(n-\nu)\times (n-\nu)}$, and $P_D\in\mathbb{R}^{\nu\times \nu}$ is a diagonal matrix. For notational clarity, we recall that $n-\nu=\tau+l+\widetilde{m}$. Moreover, we can write $A=A_\phi\widetilde{W}_0$, by defining $$A_\phi=\left[\begin{matrix}
		0_{\tau,l+\widetilde{m}}&I_{\tau}&0_{\tau,l}\\
		0_{l,l+\widetilde{m}}&0_{l,\tau}&I_l\\
		0_{\widetilde{m},l+\widetilde{m}}&0_{\widetilde{m},\tau}&0_{\widetilde{m},l}\\
		W_{\phi_1}&W_{\phi_2}&W_{\phi_3}
	\end{matrix}\right]\,,$$
	and 
	$$\widetilde{W}_0=\left[\begin{matrix}
		I_{l+\widetilde{m}}&0_{l+\widetilde{m},\tau}&0_{l+\widetilde{m},\nu}\\
		0_{\tau,l+\widetilde{m}}&I_{\tau}&0_{\tau,\nu}\\
		0_{l,l+\widetilde{m}}&0_{l,\tau}&W_0
	\end{matrix}\right]\,.$$ Furthermore, we can define $\widetilde{A}_\phi=WA_\phi=\left[\begin{matrix}
		Q\\
		L_pW_\phi
	\end{matrix}\right]$, where $W_\phi=\begin{bmatrix}W_{\phi_1}&W_{\phi_2}&W_{\phi_3}\end{bmatrix}\,,$ and $$Q=\left[\begin{matrix}
		0_{\tau,l+\widetilde{m}}&I_{\tau}&0_{\tau,l}\\
		0_{l,l+\widetilde{m}}&0_{l,\tau}&I_l\\
		0_{\widetilde{m},l+\widetilde{m}}&0_{\widetilde{m},\tau}&0_{\widetilde{m},l}
	\end{matrix}\right]\,.$$ 
	\\
	Therefore, the assumptions of Theorem \ref{th2} can be rewritten as follows
	\begin{align}
	\widetilde{W}_0^T\widetilde{A}_\phi^TP\widetilde{A}_\phi\widetilde{W}_0-P&\prec0\,,\notag\\
	\widetilde{W}_0^T(Q^T\widetilde{P}Q+L_p^2W_\phi^TP_DW_\phi)\widetilde{W}_0-P&\prec0\,,\notag\\
	\label{eqp4}
	\widetilde{W}_0^TQ^T\widetilde{P}Q\widetilde{W}_0-P+L_p^2\widetilde{W}_0^TW_\phi^TP_DW_\phi\widetilde{W}_0&\prec0\,.
	\end{align}
	Now, let us choose a block diagonal $\widetilde{P}$, i.e., $\widetilde{P}=\widetilde{P}^T=\text{diag}(\widetilde{P}_1,\widetilde{P}_2,\widetilde{P}_3)$, where $\widetilde{P}_1\in\mathbb{R}^{\tau\times \tau}$, $\widetilde{P}_2\in\mathbb{R}^{l\times l}$, and \mbox{$\widetilde{P}_3\in\mathbb{R}^{\widetilde{m}\times \widetilde{m}}$}. With this choice, we can compute
	\begin{align*}
		\widetilde{W}_0^TQ^T\widetilde{P}Q\widetilde{W}_0=\widetilde{W}_0^T\left[\begin{matrix}
			0_{l+\widetilde{m},l+\widetilde{m}}&0_{l+\widetilde{m},\tau}&0_{l+\widetilde{m},l}\\
			0_{\tau,l+\widetilde{m}}&\widetilde{P}_1&0_{\tau,l}\\
			0_{l,l+\widetilde{m}}&0_{l,\tau}&\widetilde{P}_2
		\end{matrix}\right]\widetilde{W}_0=\left[\begin{matrix}
			\widetilde{P}_0&0_{n-\nu,\nu}\\
			0_{\nu,n-\nu}&W_0^T\widetilde{P}_2W_0
		\end{matrix}\right]\,,
	\end{align*}
	where $\widetilde{P}_0=\left[\begin{matrix}
		0_{l+\widetilde{m},l+\widetilde{m}}&0_{l+\widetilde{m},\tau}\\
		0_{\tau,l+\widetilde{m}}&\widetilde{P}_1\\
	\end{matrix}\right]$. Thus, we have that
	\begin{align*}
		\widetilde{W}_0^TQ^T\widetilde{P}Q\widetilde{W}_0-P=\left[\begin{matrix}
			\widetilde{P}_0-\widetilde{P}&0_{n-\nu,\nu}\\
			0_{\nu,n-\nu}&W_0^T\widetilde{P}_2W_0-P_D
		\end{matrix}\right]\,.
	\end{align*}
	Then, let us choose $$\widetilde{P}_1=\alpha\cdot\text{diag}(I_{l+\widetilde{m}},2I_{l+\widetilde{m}},\dots,(N-2)I_{l+\widetilde{m}},(N-1)I_{\widetilde{m}})\,,$$ $\widetilde{P}_2=\alpha(N-1)I_{l}$, $\widetilde{P}_3=\alpha NI_{\widetilde{m}}$, for any scalar $\alpha>0$, and $P_D=\beta I_{\nu}$, for any scalar $\beta>0$. Hence, according to the previous choices, we can rewrite \eqref{eqp4} as
	\begin{align*}
		\left[\begin{matrix}
			-\alpha I_{n-\nu}&0_{n-\nu,\nu}\\
			0_{\nu,n-\nu}&\alpha(N-1)W_0^TW_0-\beta I_{\nu}
		\end{matrix}\right]+\beta L_p^2\widetilde{W}_0^TW_\phi^TW_\phi\widetilde{W}_0\prec0\,.
	\end{align*}
	By adding and subtracting $\gamma\widetilde{W}_0^T\widetilde{W}_0$, for any scalar $\gamma>0$, the condition \eqref{eqp4} becomes
	\begin{align}
		\label{eq2p4}
		\left[\begin{matrix}
			-\alpha I_{n-\nu}&0_{n-\nu,\nu}\\
			0_{\nu,n-\nu}&\alpha(N-1)W_0^TW_0-\beta I_{\nu}
		\end{matrix}\right]+\gamma\widetilde{W}_0^T\widetilde{W}_0+\widetilde{W}_0^T(\beta L_p^2W_\phi^TW_\phi-\gamma I_{(l+\widetilde{m})N})\widetilde{W}_0\prec0\,.
	\end{align}
	Now, note that if both
	\begin{align}
		\label{eq3p4}
		\left[\begin{matrix}
			(\gamma-\alpha) I_{n-\nu}&0_{n-\nu,\nu}\\
			0_{\nu,n-\nu}&(\alpha(N-1)+\gamma)W_0^TW_0-\beta I_{\nu}
		\end{matrix}\right]&\prec0\\
		\label{eq4p4}
		\beta L_p^2W_\phi^TW_\phi-\gamma I_{(l+\widetilde{m})N}&\prec0
	\end{align}
	are fulfilled, then \eqref{eq2p4} certainly holds. Therefore, we want to prove that there exist $\alpha>0$, $\beta>0$, and $\gamma>0$ such that \eqref{eq3p4} and \eqref{eq4p4} hold provided that the assumption of Proposition \ref{deltaISSNNARX} is fulfilled. Firstly, note that \eqref{eq3p4} holds if and only if 
	\begin{align}
		\label{cnd1p4}
		\gamma&<\alpha\,,\\
		\label{cnd2p4}
		\|W_0\|^2=\lambda_{max}(W_0^TW_0)&<\frac{\beta}{\alpha(N-1)+\gamma}\,.
	\end{align}
	Secondly, \eqref{eq4p4} holds if and only if 
	\begin{align}
		\label{cnd3p4}
		\|W_\phi\|^2=\lambda_{max}(W_\phi^TW_\phi)&<\frac{\gamma}{\beta L_p^2}\,.
	\end{align}
	Furthermore, if \eqref{cnd1p4} holds, we have that \eqref{cnd2p4} is certainly fulfilled if 
	\begin{align}
		\label{cnd4p4}
		\|W_0\|^2&<\frac{\beta}{\alpha N}
	\end{align}
	holds. Hence, if conditions \eqref{cnd1p4}, \eqref{cnd3p4}, and \eqref{cnd4p4} are satisfied, we have that \eqref{eq2p4} holds, implying the fulfilment of the assumptions of Theorem \ref{th2}. Finally, to have \eqref{cnd1p4}, \eqref{cnd3p4}, and \eqref{cnd4p4} fulfilled, we have to prove the existence of positive $\alpha$, $\beta$, and $\gamma$ such that 
	\begin{align}
		\label{cnd5p4}
		\|W_\phi\|^2L_p^2<\frac{\gamma}{\beta}<\frac{\alpha}{\beta}<\frac{1}{\|W_0\|^2N}\,.
	\end{align}
	In particular, in view of the assumption of Proposition \ref{deltaISSNNARX}, i.e., $\|W_0\|\|W_\phi\|<\frac{1}{L_p\sqrt{N}}$, it follows that $\|W_\phi\|^2L_p^2<\frac{1}{\|W_0\|^2N}$. Thus, it is always possible to find positive real numbers $\alpha$, $\beta$, and $\gamma$ such that \eqref{cnd5p4} holds thanks to the completeness axiom of real numbers. This concludes the proof.
\end{proof}
\renewcommand*{\proofname}{Proof of Proposition \ref{prop:china2}}
\begin{proof}
	Firstly, if $E=0_{n,n}$ and $o_i = 1$ $\forall i \in \mathcal{L}$, we have that $\hat{M}=W|A|$, where $A=O\hat{A}$. Accordingly, the assumption of Proposition \ref{prop:china} requires the existence of a matrix $P=P^T \succ 0$ such that
	\begin{align}\label{eq:cond_china2}
		|A|^T \,W^T P\, W\,|A| \,- \,P \prec 0\,.
	\end{align}
	This condition represents also the stability condition for
	\begin{align}\label{eq:pos_system}
		\tilde{x}(k+1) = W\,|A| \,\tilde{x}(k)\,,
	\end{align}
	which is a discrete-time linear positive system. Hence, as stated in \cite[Theorem 15]{farina2000positive}, system \eqref{eq:pos_system} is asymptotically stable if and only if there exists a diagonal $P\succ 0$ fulfilling condition \eqref{eq:cond_china2}.
	
	Therefore, let us consider a matrix $P=\text{diag}(p_1,\dots,p_n) \in \mathbb{R}^{n \times n}$, with $p_i>0$. For any generic vector $v\in \mathbb{R}^n$, $v\neq0_{n,1}$, condition \eqref{eq:cond_china2} implies that 
	\begin{align*}
		v^T(|A|^T \,W^T P\, W\,|A| \,- \,P)\,v < 0\,,
	\end{align*}
	which becomes 
	\begin{align*}
		(i)\qquad\sum_{j=1}^n p_j L_{pj}^2 \bigg  (\sum_{i=1}^n |a_{ji}| v_i\bigg)^2 \;-\; \sum_{j=1}^n p_j v_j^2 <0\,.
	\end{align*}
	On the other hand, following the same reasoning, for all \mbox{$\tilde{v}\in\mathbb{R}^n$}, $\tilde{v}\neq0_{n,1}$, condition \eqref{lyaplin} in Theorem \ref{th2} with a diagonal $P$ is equivalent to
	\begin{align*}
		(ii)\qquad	\sum_{j=1}^n p_j L_{pj}^2 \bigg  (\sum_{i=1}^n a_{ji} \tilde{v}_i\bigg)^2 \;-\; \sum_{j=1}^n p_j \tilde{v}_j^2 <0\,.
	\end{align*}
	
	We want to show now that, if $(i)$ holds for any $v\in \mathbb{R}^n$, $v\neq0_{n,1}$, it follows that $(ii)$ holds for any $\tilde{v}\in\mathbb{R}^n$, $\tilde{v}\neq0_{n,1}$.
	
	Let us consider a generic $\tilde{v}\in \mathbb{R}^n$, $\tilde{v}\neq0_{n,1}$. It follows that $v=|\tilde{v}|$ satisfies $(i)$ by assumption. Also, note that
	\begin{align*}
		\bigg|\sum_{i=1}^n a_{ji}\, \tilde{v}_i\bigg| \leq \sum_{i=1}^n \Big|a_{ji}\, \tilde{v}_i \Big| = \sum_{i=1}^n |a_{ji}|\, | \tilde{v}_i|\, = \sum_{i=1}^n |a_{ji}|\, v_i \,,
	\end{align*}
	which implies that
	\begin{align}\label{eq:proof_china_1}
		\sum_{j=1}^n p_j L_{pj}^2 \bigg  (\sum_{i=1}^n a_{ji}\, \tilde{v}_i\bigg)^2 \leq 	\;\,\sum_{j=1}^n p_j L_{pj}^2 \bigg  (\sum_{i=1}^n |a_{ji}|\, v_i\bigg)^2 \,.
	\end{align}
	Moreover, note that  
	\begin{align}\label{eq:proof_china_2}
		\sum_{j=1}^n p_j \tilde{v}_j^2\, = \sum_{j=1}^n p_j |\tilde{v}_j|^2\, =  \sum_{j=1}^n p_j v_j^2 \,.
	\end{align}
	Equations \eqref{eq:proof_china_1} and \eqref{eq:proof_china_2} imply that, since $v=|\tilde{v}|$ verifies $(i)$, then $\tilde{v}$ will satisfy $(ii)$. This means that $(i)\implies(ii)$, which concludes the proof.
\end{proof}
\renewcommand*{\proofname}{Proof of Proposition \ref{cas}}
\begin{proof}
	The proof is carried out with reference to the series of two systems \eqref{nonlinclassall}, since the generalization to a generic number of systems follows straightforwardly by iterating the procedure. 
	
	Let us consider two systems in the class \eqref{nonlinclassall}, using the subscripts $1$ (system upstream) and $2$ (system downstream) to denote the corresponding matrices, respectively. According to Proposition \ref{ffp}, the series of the two systems is in the class~\eqref{nonlinclassall}, where $$A\!=\!\left[\begin{matrix}
		A_1&0_{n_1,n_2}\\
		B_2C_1&A_2
	\end{matrix}\right]\,,\ \ x(k)\!=\!\begin{bmatrix}x_1(k)\\x_2(k)\end{bmatrix}\,,\ \ f(\cdot)\!=\!\begin{bmatrix}f_{s_1}(\cdot)\\f_{s_2}(\cdot)\end{bmatrix},$$ for the overall series system. Accordingly, we have that 
	\begin{align*}
		\widetilde{A}=WA=\left[\begin{matrix}
			W_1&0_{n_1,n_2}\\
			0_{n_2,n_1}&W_2
		\end{matrix}\right]\left[\begin{matrix}
			A_1&0_{n_1,n_2}\\
			B_2C_1&A_2
		\end{matrix}\right]=\left[\begin{matrix}
			W_1A_1&0_{n_1,n_2}\\
			W_2B_2C_1&W_2A_2
		\end{matrix}\right]=\left[\begin{matrix}
			\widetilde{A}_1&0_{n_1,n_2}\\
			\widetilde{A}_{21}&\widetilde{A}_2
		\end{matrix}\right]\,.
	\end{align*}
	
	By assumption, the two systems fulfil the assumptions of Theorem \ref{th2}, i.e., $\exists P_1=P_1^T\succ0$ and $\exists P_2=P_2^T\succ0$ with the required structures such that \begin{align}
		\widetilde{A}_1^TP_1\widetilde{A}_1-P_1&\prec0\,, \label{cnd1}\\
		\widetilde{A}_2^TP_2\widetilde{A}_2-P_2&\prec0\,. \label{cnd2}
	\end{align}
	From \eqref{cnd2}, for any scalar $\rho>0$ , we have that $$\widetilde{A}_1^T\rho P_1\widetilde{A}_1-\rho P_1\prec0\,.$$
	
	Now, we want to prove that $\exists P=P^T\succ0$ with the required structure such that $\widetilde{A}^TP\widetilde{A}-P\prec0$. Firstly, we take $$P=P^T=\left[\begin{matrix}
		\rho P_1&0_{n_1,n_2}\\
		0_{n_2,n_1}&P_2
	\end{matrix}\right]\succ0\,,$$ which has the structure required by Theorem \ref{th2} on the basis of system nonlinearities. Then,
	\begin{align*}
		\begin{split}
			\widetilde{A}^TP\widetilde{A}-P=\left[\begin{matrix}
				\widetilde{A}_1^T\rho P_1\widetilde{A}_1+\widetilde{A}_{21}^TP_2\widetilde{A}_{21}-\rho P_1&\widetilde{A}_{21}^TP_2\widetilde{A}_2\\
				\widetilde{A}_2^TP_2\widetilde{A}_{21}&\widetilde{A}_2^TP_2\widetilde{A}_2-P_2
			\end{matrix}\right]\,.
		\end{split}
	\end{align*}
	We define a generic vector $v=\begin{bmatrix}v_1^T&v_2^T\end{bmatrix}^T$, where \mbox{$v_1\in\mathbb{R}^{n_1}$} and $v_2\in\mathbb{R}^{n_2}$. Therefore, we want to prove that $v^T(\widetilde{A}^TP\widetilde{A}-P)v<0$ for any $v\neq0_{n_1+n_2,1}$. Hence,
	\begin{align*}
		v^T(\widetilde{A}^TP\widetilde{A}-P)v=v_1^T(\widetilde{A}_1^T\rho P_1\widetilde{A}_1+\widetilde{A}_{21}^TP_2\widetilde{A}_{21}-\rho P_1)v_1+v_2^T(\widetilde{A}_2^TP_2\widetilde{A}_2-P_2)v_2+2v_2^T\widetilde{A}_2^TP_2\widetilde{A}_{21}v_1\,.
	\end{align*}
	In view of Property \ref{Prop1}, for any $\tau\neq0$, we have that
	\begin{align*}
		v^T(\widetilde{A}^TP\widetilde{A}-P)v\leq v_1^T(\widetilde{A}_1^T\rho P_1\widetilde{A}_1-\rho P_1+(1+\frac{1}{\tau^2})\widetilde{A}_{21}^TP_2\widetilde{A}_{21})v_1+v_2^T((1+\tau^2)\widetilde{A}_2^TP_2\widetilde{A}_2-P_2)v_2\,.
	\end{align*}
	Note that $(1+\tau^2)\widetilde{A}_2^TP_2\widetilde{A}_2-P_2\prec0$ in view of~\eqref{cnd2} for a sufficiently small $\tau$. Moreover, $A_\tau^*=(1+\frac{1}{\tau^2})\widetilde{A}_{21}^TP_2\widetilde{A}_{21}\preceq\lambda_{max}(A_\tau^*)I_{n_1}$, where $\lambda_{max}(A_\tau^*)\geq0$. Furthermore, $\widetilde{A}_1^T\rho P_1\widetilde{A}_1-\rho P_1\preceq \rho\lambda_{max}(A_1^*)I_{n_1}$, where $A_1^*=\widetilde{A}_1^TP_1\widetilde{A}_1-P_1$ and $\lambda_{max}(A_1^*)<0$ from \eqref{cnd1}. Hence,
	\begin{align*}
		v^T(\widetilde{A}^TP\widetilde{A}-P)v\leq v_1^T(\lambda_{max}(A_\tau^*)+\rho\lambda_{max}(A_1^*))I_{n_1}v_1+v_2^T((1+\tau^2)\widetilde{A}_2^TP_2\widetilde{A}_2-P_2)v_2<0
	\end{align*}
	for any $v\neq0_{n_1+n_2,1}$, since $\lambda_{max}(A_\tau^*)+\rho\lambda_{max}(A_1^*)<0$ for a sufficiently large $\rho$. This concludes the proof.
\end{proof}

%
%
%
%
%

\bibliographystyle{ieeetr}
\bibliography{Bibliografia.bib}

\end{document}